\documentclass[a4paper, twocolumn]{article}

\usepackage{authblk}

\usepackage{cite}
\usepackage{amsmath}
\usepackage{amsfonts}
\usepackage{amsthm}
\usepackage{amssymb}
\usepackage[utf8]{inputenc}
\usepackage{url}
\usepackage{hyperref}
\usepackage{graphicx}
\usepackage{physics}
\usepackage{subcaption}
\usepackage{indentfirst}
\usepackage{nccmath}

\usepackage{blindtext}

\usepackage{abstract}

\theoremstyle{plain}
\newtheorem*{proposition}{Proposition}

\makeatletter
\let\@fnsymbol\@arabic
\makeatother

\begin{document}
\date{\today}

\title{\textbf{Denoising quantum states with Quantum Autoencoders - Theory and Applications}}
\author[ ]{Tom Achache$^{\dag,\:}$\thanks{\texttt{tna2111@columbia.edu}}}
\author[$\dag$, $\star$]{Lior Horesh}
\author[$\dag$, $\star$]{John Smolin}
\affil[ ]{${}^\dag$Columbia University $\qquad$ ${}^\star$IBM Research}

\twocolumn[
  \begin{@twocolumnfalse}
\maketitle

\begin{abstract}
We implement a Quantum Autoencoder (QAE) as a quantum circuit capable of correcting Greenberger–Horne–Zeilinger (GHZ) states subject to various noisy quantum channels : the bit-flip channel and the more general quantum depolarizing channel. The QAE shows particularly interesting results, as it enables to perform an almost perfect reconstruction of noisy states, but can also, more surprisingly, act as a generative model to create noise-free GHZ states. Finally, we detail a useful application of QAEs : Quantum Secret Sharing (QSS). We analyze how noise corrupts QSS, causing it to fail, and show how the QAE allows the QSS protocol to succeed even in the presence of noise.

\bigskip

\end{abstract}

\end{@twocolumnfalse}
]

\saythanks

\section{Introduction}

Autoencoders are a powerful type of Neural Networks that can realize compact representations of data, and are often used as a denoising method \cite{autoencoders}. Following this idea, Quantum Autoencoders (QAEs) have recently emerged as a possible solution to tackle the crippling quantum noise problem, which is ubiquitous in real-device quantum computations. Here, we choose to focus on QAEs as proposed in \cite{bondarenko2019quantum}, which enable an almost perfect reconstruction of Greenberger–Horne–Zeilinger (GHZ) states subject to random bit-flips and small unitary noise. A $m$-qubit GHZ state is defined by
\begin{align}
\frac{1}{\sqrt{2}}( \ket{0}^{\bigotimes m}+\ket{1}^{\bigotimes m}).
\end{align}

Studying GHZ states is of particular relevance as they exhibit long-range entanglement, and are used in a wide range of applications. 
\medskip

Noticeably, the circuit depth of the QAEs considered in this paper is very shallow (for instance, the QAE denoising 3-qubit GHZ states only requires 4 qubits), making
them easy to implement on currently existing quantum devices.

These QAEs rely on the Quantum Neural Networks (QNNs) architecture described in \cite{beer2019efficient}, and are trained in an unsupervised learning fashion \cite{sents2019unsupervised}. While \cite{beer2019efficient} already provides an implementation of QNNs in Matlab \cite{deepqnn}, we choose here to implement the QAEs in a pure quantum framework, using the Qiskit \cite{qiskit} library. 

In this study, we propose an end-to-end quantum (circuit) encoding of all the algorithmic steps associated with QAEs training and testing. Further, we discuss important implementation and training details accompanied by extended results of the QAEs. Lastly, we perform a thorough mathematical analysis of the impact of noise in Quantum Secret Sharing (QSS) and show how to counter it. Hence, this work covers both theoretical and experimental aspects of QAEs.
\medskip

This paper is organized as follows. Section \ref{Implementation} addresses the theoretical background regarding the implementation and training of QAEs on quantum circuits. Section \ref{Results} presents numerical results pertaining to the performance of the QAEs when denoising corrupted states subject to different kinds of noise. Eventually, in section \ref{QSS} we discuss an application of QAEs in QSS, and demonstrate their significance in that context.

\section{Algorithmic Description} \label{Implementation}
This section summarizes the theory for implementing a QNN on a quantum computer, with particular attention to the specific behavior of QAEs. The critical steps of the computation are described below, but the reader interested in the details of the calculus is referred to \cite{beer2019efficient}.

\subsection{Circuit design}

The QAE used in this study was selected from \cite{beer2019efficient} for its appealing and practical property of being able to be trained and tested using a single quantum circuit and a minimal number of qubits. We denote by $[m_1, \dots, m_\ell]$ a QNN with $\ell$ layers, each layer having size $m_i, 1\leq i \leq \ell$ (thus the input of the QNN is a quantum state of $m_1$ qubits). For such a QNN, the quantum circuit is made up of $Q$ qubits, where 
\begin{align}
    &Q = 1 + m_1 + w \\
    \text{ with } &w = \max\limits_{1 \leq i \leq \ell - 1}(m_i + m_{i+1}).
\end{align}

$w$ is the width of the QNN, and the other bits provide a measure of the fidelity (see Eq. \eqref{eq:fid}) between the results of the QNN and the target states during the training phase. The fidelity between a state $\ket{\phi}$ and a density matrix $\rho$ is defined as
\begin{align}
    F(\ket{\phi}, \rho) = \ev{\rho}{\phi}.
    \label{eq:fid}
\end{align}

It should be underlined that once the training has been completed, we only use the QAE, which has $w$ qubits. An example of the complete circuit is represented in Figure \ref{fig:QNN}. We divided it in four parts, separated with vertical grey lines ("barriers"). The first part of the network represents the states' preparation, the second part contains the QAE itself, the third part allows us to compute the fidelity between the output of the QAE and the target state, as demonstrated in \cite{beer2019efficient}, and the last part is the measurement. If we denote $\rho$ the density matrix of the state output by the QAE, $\ket{\phi}$ the target state, and $p_0$ the probability of getting $0$ when doing the measurement, we then have 
\begin{align}
p_0 = \frac{1}{2}(1 + F(\ket{\phi}, \rho)),
\label{eq:p0}
\end{align}
where $F$ is the fidelity function defined in Eq. \eqref{eq:fid}.

\begin{figure*}
\centering
\includegraphics[scale = 0.4]{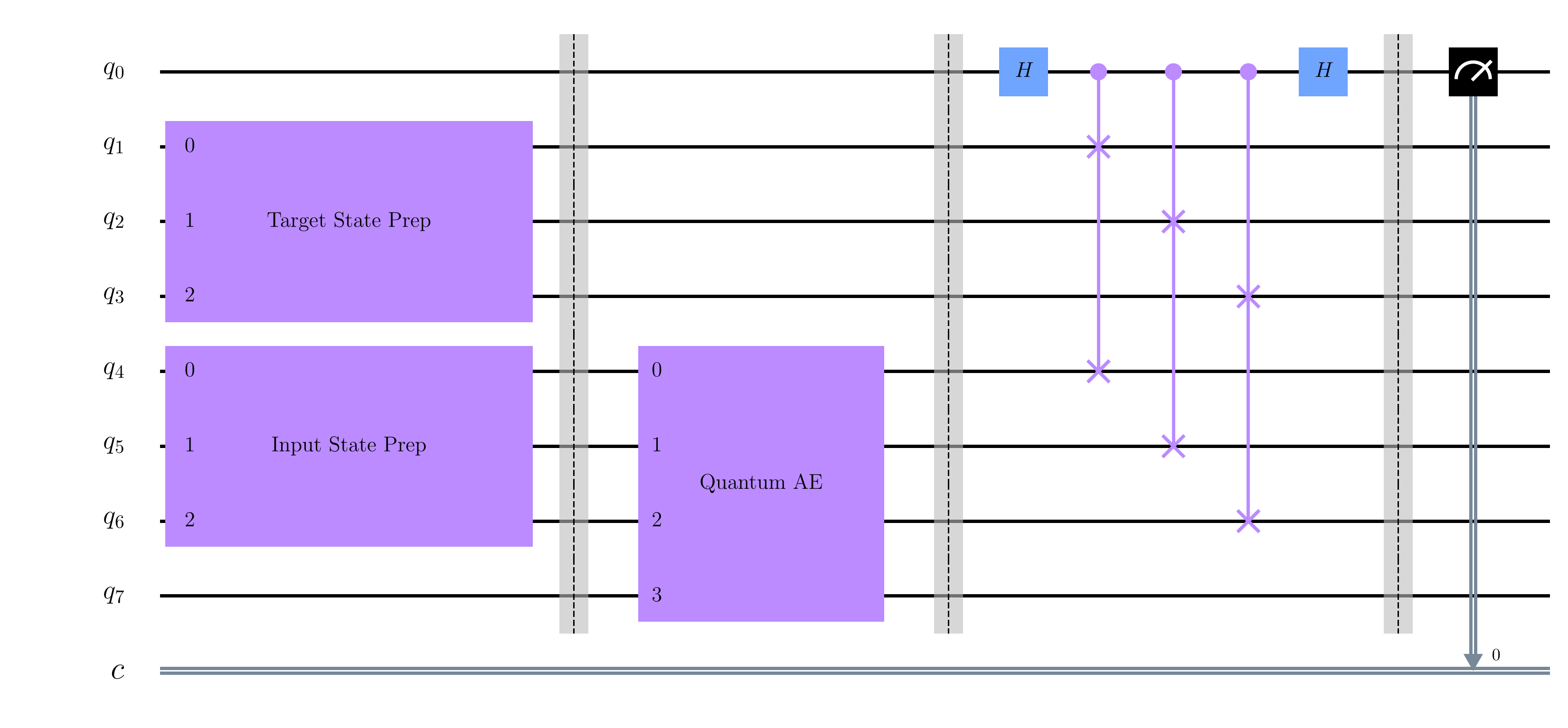}
\caption{Circuit for the training of a [3,1,3] QAE. Here, $w = 4$ and $Q = 8$.}
\label{fig:QNN}
\end{figure*}

The QAE circuit is represented in Figure \ref{fig:QAE}. It has two internal layers, separated by a barrier.

\begin{figure}[ht]
\centering
\includegraphics[scale = 0.3]{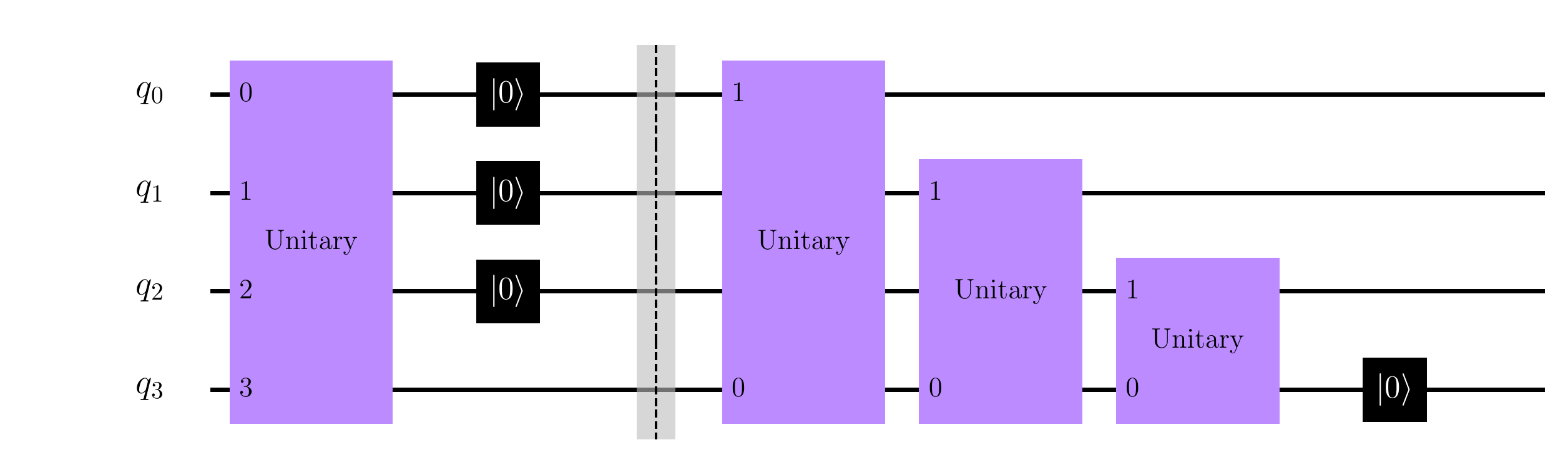}
\caption{Circuit of a [3,1,3] QAE}
\label{fig:QAE}
\end{figure}

As explained in \cite{beer2019efficient}, a QNN is composed of layers of qubits (replacing the neurons in classical NNs), connected by unitary transformations (replacing the weights). Each unitary transformation acts on all the qubits of the input layer and one qubit of the output layer. This is equivalent to classical NNs where each neuron at a certain layer is connected to all the neurons of the previous one. Thus, there are as many unitary transformations connecting two layers as there are qubits in the output layer. Moreover, each unitary transformation $U$ is defined by
\begin{align}
U = e^{iK} \text{ with } K = \sum\limits_{\sigma \in P^{\otimes (m + 1)}} k_\sigma \cdot \sigma,
\label{eq:unitaries}
\end{align}
where the $k_\sigma$ are the parameters to be learned, $m$ is the size of the input layer, $P = \{I, X, Y, Z \}$ is the set of Pauli matrices and the identity, and $P^{\otimes j}$ denotes the set of all possible tensor products of length $j$ between the elements of $P$. For instance, $P^{\otimes 2} = \{II, IX, IY, IZ, XI, XX, XY, XZ, \ldots \}$. Since $P$ forms a basis for the real vector space of $2 \times 2$ Hermitian matrices, $P^{\otimes (m + 1)}$ naturally forms a basis for the real vector space of $2^{m+1} \times 2^{m+1}$ Hermitian matrices. Hence, $K$ is uniquely defined by its coefficients $k_\sigma$.

As a reminder, the expressions of the Pauli matrices are
\begin{align}
X = \begin{pmatrix}
0 & 1 \\ 1 & 0
\end{pmatrix}, Y = \begin{pmatrix}
0 & -i\\ i & 0
\end{pmatrix}, Z = \begin{pmatrix}
1 & 0 \\ 0 & -1
\end{pmatrix}
\end{align}

Thus, between two layers of size $m_i$ and $m_{i+1}$, we will have $m_{i+1}$ unitary transformations, each having $4^{m_i + 1}$ coefficients. Hence, for a circuit of shape $[m_1,\ldots, m_\ell]$, the total number of coefficients trained is 
\begin{align}
\sum\limits_{i = 1}^{\ell-1} m_{i+1} \cdot 4^{m_i + 1}.
\label{eq:len}
\end{align}

In order to use the minimum possible amount of qubits for the QAE circuit, we can reuse the qubits of the precedent layers by resetting them, as illustrated in Figure \ref{fig:QAE} (the $\ket{0}$'s). In practice, at all times during training, we only need the qubits representing two consecutive layers. So we can use a fixed number of qubits by resetting the ones used in the past layers. However, one has to be particularly careful at the end of the circuit to identify the qubits where the states is actually encoded.

\subsection{Training}

The aim of the training part is to maximize the fidelity between the output of the QAE and the actual state(s) we are trying to reconstruct (for instance, the noise-free GHZ states).

As in a regular Neural Network, we have $N$ training pairs, which here consist in two circuits preparing the input and target states (see Figure \ref{fig:QNN}). Since our aim is to denoise GHZ states, it is not possible to set them as the targets of our network, because that would imply that actual access to them is viable. Rather, we train our QAE in an unsupervised way, similarly to the training of classical AEs. However, classical AEs are trained with pairs $(x,x)$, i.e. where the input and the target are identical. This method does not apply here, as we do not have access to the specific noise affecting a state. Therefore, we train our QAE on pairs $(x,y)$, where $x$ and $y$ are drawn from the same noisy distribution.
\medskip

The cost function $C$ is a function of the vector $\kappa$ of all the parameters $k_\sigma$. We denote by $\rho_i^\kappa$ the output of the QAE with parameters $\kappa$, when fed with the $i$-th training state $\ket{\phi}_i$. We have 
\begin{align}
C(\kappa) = \frac{1}{N} \sum\limits_{i=1}^N F(\ket{\phi}_i, \rho_i^\kappa).
\end{align}

The network is trained using a gradient ascent method. To approximate the gradient of the cost function, we use the finite difference method, due to its simplicity. However, this method can be rather unstable numerically in the case of noisy functions, and one might want to consider other Derivative Free Optimization methods instead \cite{conn2009introduction}. 

Let $\epsilon >0$, for every parameter $k_\sigma$, if $\epsilon$ is small enough,
\begin{align}
\frac{\partial C}{\partial k_\sigma} (\kappa) \simeq \frac{C(\kappa + \epsilon_{k_\sigma}) - C(\kappa)}{\epsilon},
\end{align}
where $\epsilon_{k_\sigma}$ is an indicator vector with $\epsilon$ at the position corresponding to $k_\sigma$, and $0$ otherwise. 
We can then update the parameters vector $\kappa$ with the classical rule
\begin{align}
\kappa \longleftarrow \kappa + \eta  \frac{\partial C}{\partial k_\sigma} (\kappa),
\end{align}
where $\eta$ is the learning rate.
\medskip 

Hence, we implemented the training algorithm as follows : at each epoch, we start by computing $C(\kappa)$ (with the circuit depicted in Figure \ref{fig:QNN}), and then, for each coefficient $k_\sigma$, we compute $C(\kappa + \epsilon_{k_\sigma})$. Finally, we update $\kappa$.

In terms of complexity, we need to make $N \cdot (\vert \kappa \vert + 1) \cdot S$ measurements, where $\vert \kappa \vert$ is the value defined in Eq. \eqref{eq:len} and $S$ is the number of times we run each circuit in order to estimate the probability of getting 0 as a result (and thus the fidelity, through Eq. \eqref{eq:p0}). Choosing a large $S$ may unnecessarily increase the computation time. Conversely, choosing a small $S$ may result in a poor approximation of $p_0$ (and thus of the fidelity) which will hinder the training performance. We decided to use $S=1000$ in all simulations, as it offers a sensible approximation of $p_0$ while keeping the number of measurements reasonable.

All the measurements can be done sequentially. This confers its practical interest to the QAE as it can be implemented and trained using a single and relatively small quantum circuit. However, running measurements in parallel can significantly decrease the training time of the algorithm, as each epoch can be entirely parallelized.  We included both options in our code.
\medskip

Besides, we need to tune two hyperparameters in the training phase, $\epsilon$ (the differential length) and $\eta$ (the learning rate), which can be prohibitive. However, while the success of the training phase is dependent on these two parameters, in practice $\epsilon = 0.1$ and $\eta = 1/4$ gave very good results in all cases. We also implemented Adam gradient ascent \cite{kingma2014adam} as it should be able to curb the significance of $\eta$ due to its adaptive learning rate. However, using the vanilla gradient ascent with the value of $\eta$ stated yielded more than satisfying results.

\section{Results} \label{Results}

Due to the high number of coefficients (see Eq. \eqref{eq:len}) and training states needed to properly train the QAE, the training phase is heavily time-expensive. Indeed, we were able to obtain almost-perfect results (fidelity above 0.95) on [2,1,2] and [3,1,3] QAEs, but the resources needed for the training of larger networks exceeded our capacity. But the method should work in theory, independently of the size of the network. \cite{bondarenko2019quantum} shows the results on a [4,2,1,2,4] network.

\subsection{Bit-flip channel}

We first tested our QAE on the task of denoising GHZ states submitted to random bit-flips. In our model, every bit of the state is flipped with a certain probability $p$.

We trained a $[2,1,2]$ QAE with 100 training pairs, drawn by taking GHZ states submitted to a bit-flip probability $p = 0.2$. The training curve is shown in Figure \ref{fig:train_bitflip}. At the end of each epoch we computed the fidelity on the training pairs and on validation pairs (100 pairs, drawn from the same distribution as the training pairs). As we can observe by the small discrepancies between the training and the validation sets performance in Figure \ref{fig:train_bitflip}, overfitting is not a concern here. Indeed, the validation set performance closely follows the performance of the training set, and in particular does not start decreasing at some point.

\begin{figure}[ht]
\centering
\includegraphics[scale = 0.5]{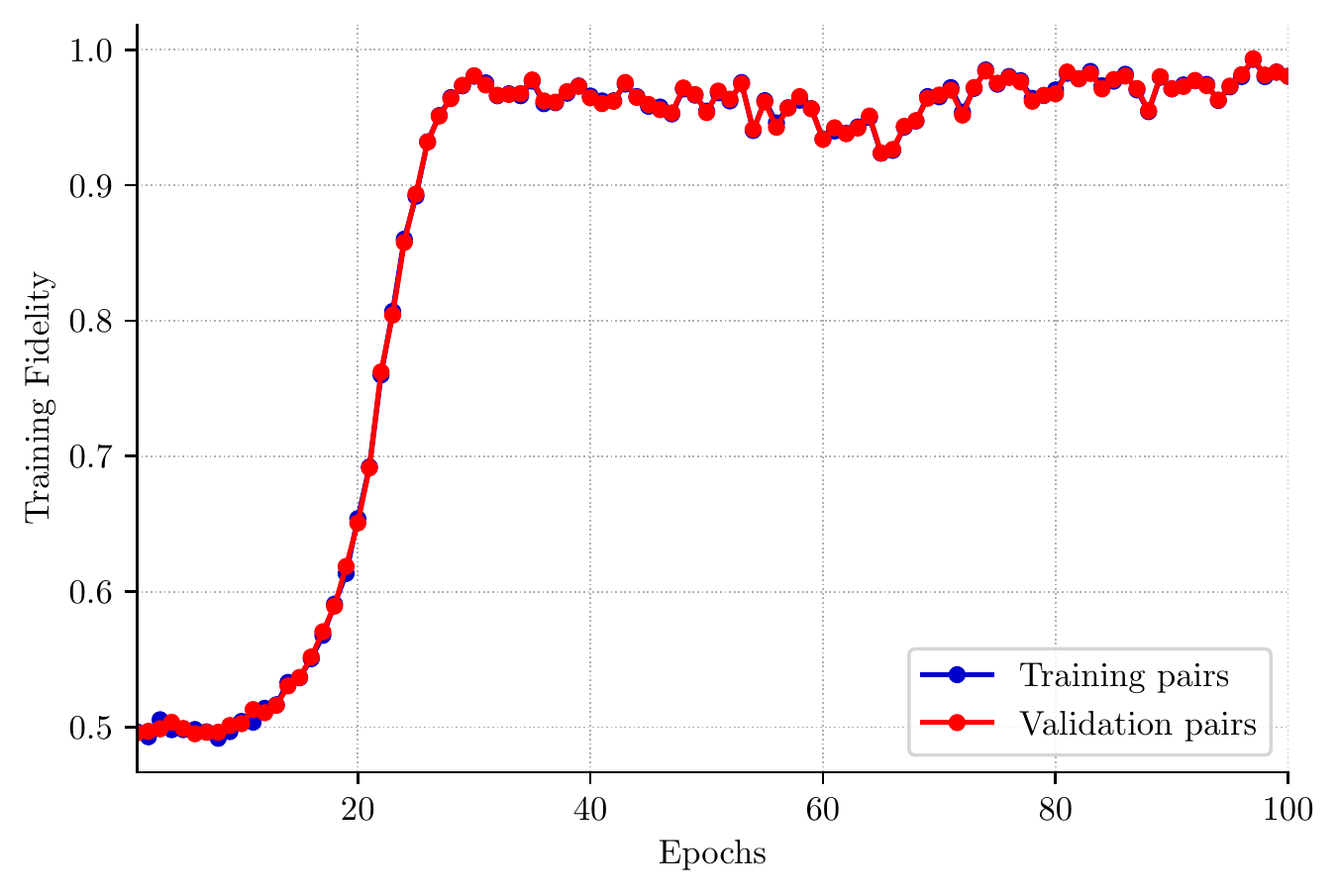}
\caption{Training fidelity of a $[2,1,2]$ QAE using 100 training pairs, for a bit-flip noise of probability $p=0.2$. We used $\epsilon = 0.1$ and $\eta = 1/4$.}
\label{fig:train_bitflip}
\end{figure}

We then tested our QAE on a test set consisting of 200 pairs drawn from the same distribution (GHZ states affected by bit-flips with probability 0.2). The average fidelity of the test set was 0.99, demonstrating that our QAE is working (almost) perfectly.
\medskip

Thereafter, we evaluated whether the QAE trained for $p=0.2$ would maintain its performance for larger bit-flip probabilities (i.e. with a stronger noise). The results are displayed in Figure \ref{fig:fid_bitflip}. The theoretical fidelity of noisy states for a probability $p$ can be computed by noting that when submitting a two-qubit GHZ state to random bit-flips, there are only two possible resulting states : 
\begin{align}
    \frac{1}{\sqrt{2}}(\ket{00} + \ket{11}) \text{ and } \frac{1}{\sqrt{2}}(\ket{01} + \ket{10}). \label{eq:bitflip}
\end{align}

The fidelities of these states with the GHZ state are respectively 1 and 0. Moreover, to get the first state (the GHZ state), either zero or both qubits have to be flipped. Hence, the theoretical fidelity is simply 
\begin{align}
(1 - p)^2 + p^2.
\label{eq:proba_stateflip}
\end{align}

\begin{figure}[ht]
\centering
\includegraphics[scale = 0.5]{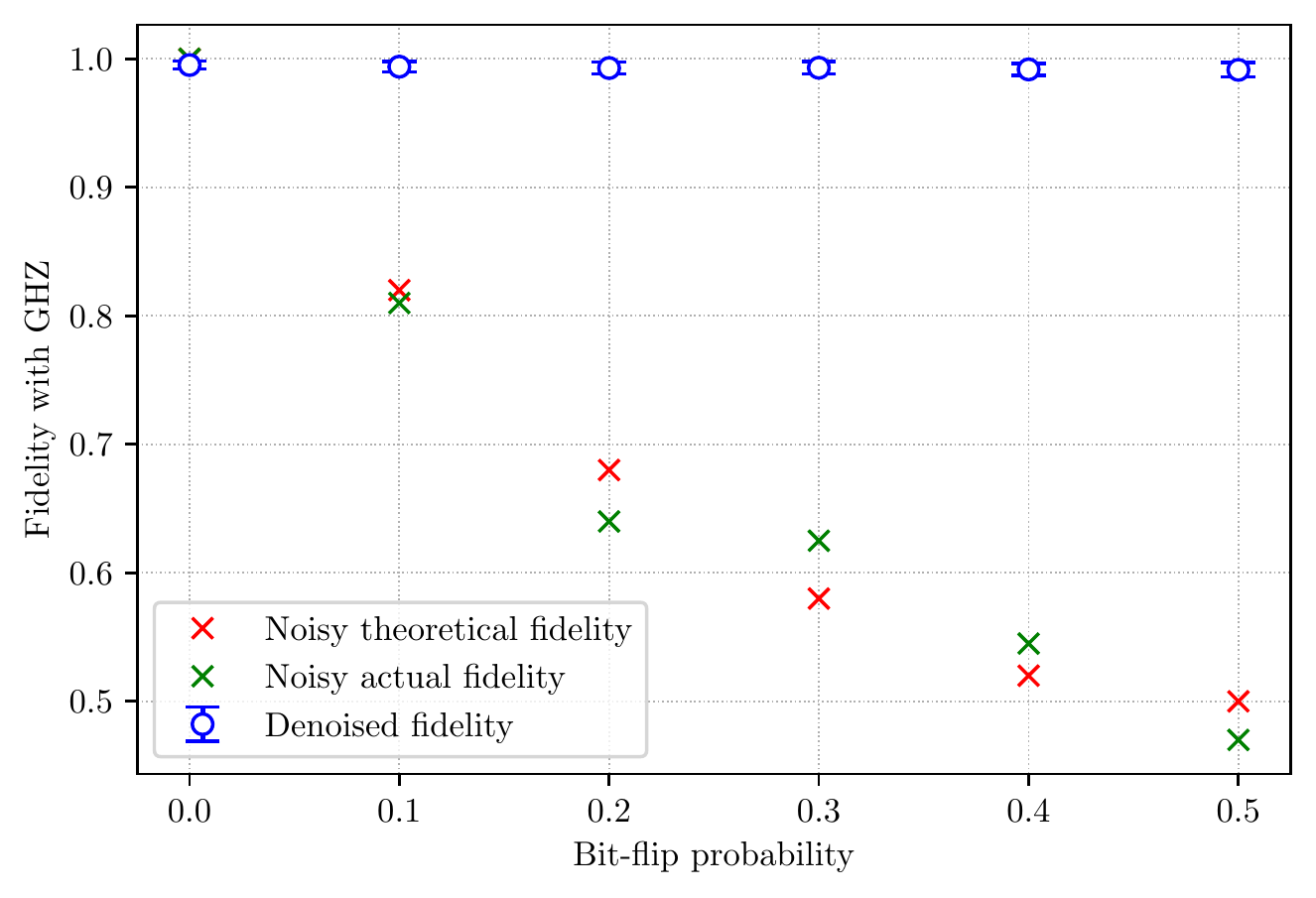}
\caption{Mean and deviation of the fidelity of the $[2,1,2]$ network, trained with $p = 0.2$, over different bit-flip probabilities. For each $p$, 200 test states were drawn. The green points are the fidelities of the test states. The red points are their theoretical fidelities.}
\label{fig:fid_bitflip}
\end{figure}

\begin{figure*}
\begin{subfigure}{.33\textwidth}
  \centering
  \includegraphics[width=0.9\textwidth]{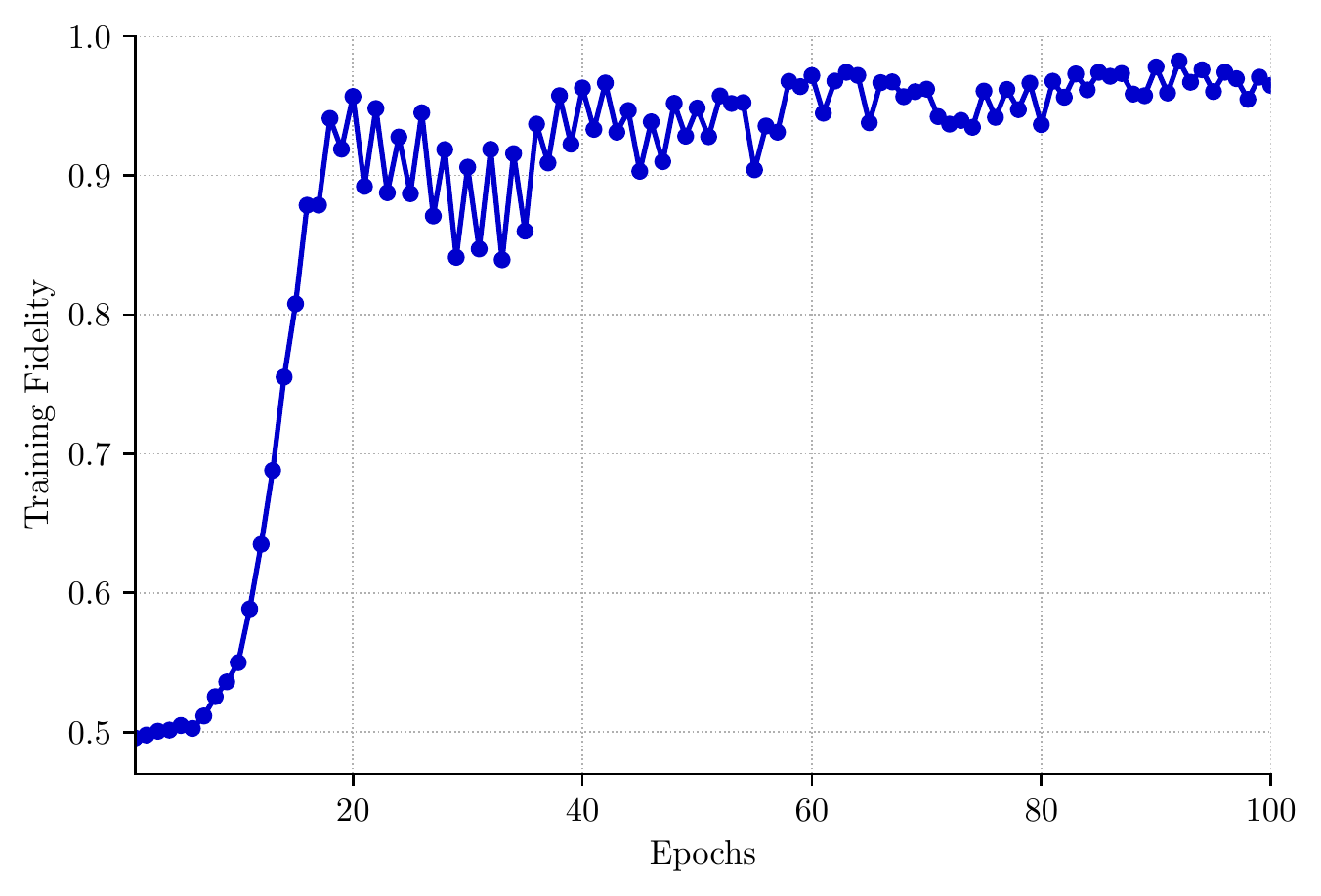}
  \caption{$p = 0.1$}
  \label{fig:a}
\end{subfigure}
\begin{subfigure}{.33\textwidth}
  \centering
  \includegraphics[width=0.9\textwidth]{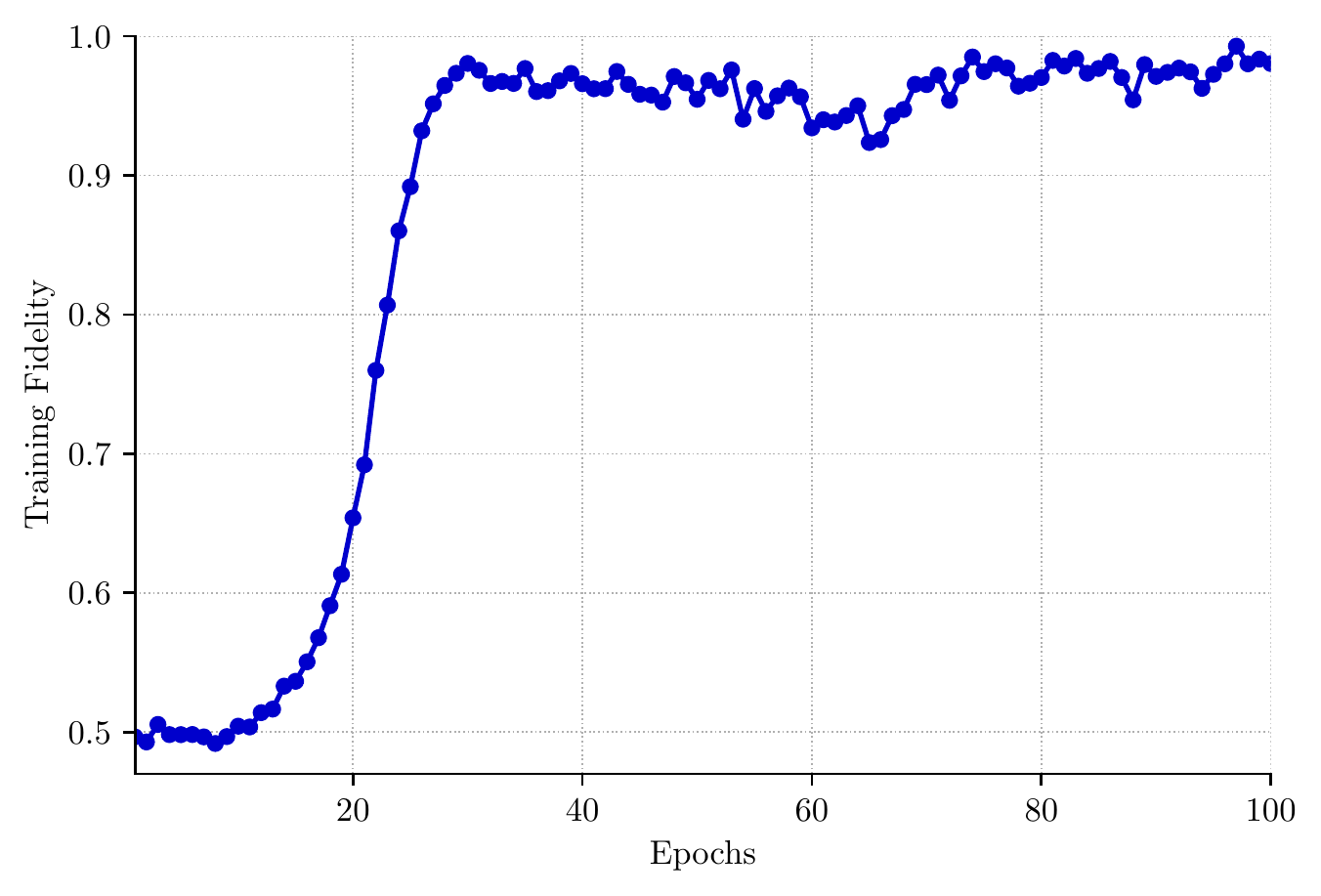}
  \caption{$p = 0.2$}
\end{subfigure}
\begin{subfigure}{.33\textwidth}
  \centering
  \includegraphics[width=0.9\textwidth]{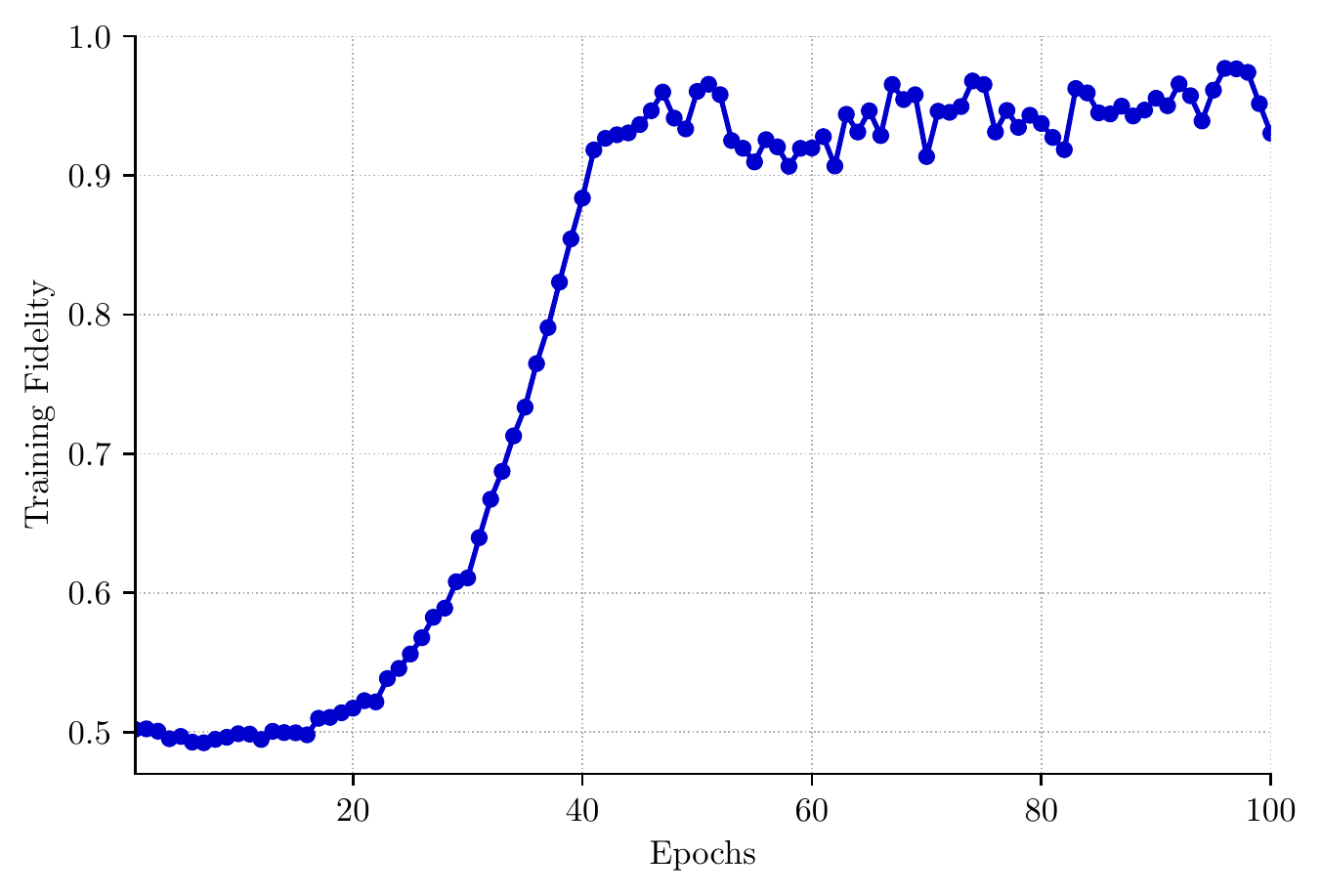}
  \caption{$p = 0.3$}
\end{subfigure}
\begin{subfigure}{.33\textwidth}
  \centering
  \includegraphics[width=0.9\textwidth]{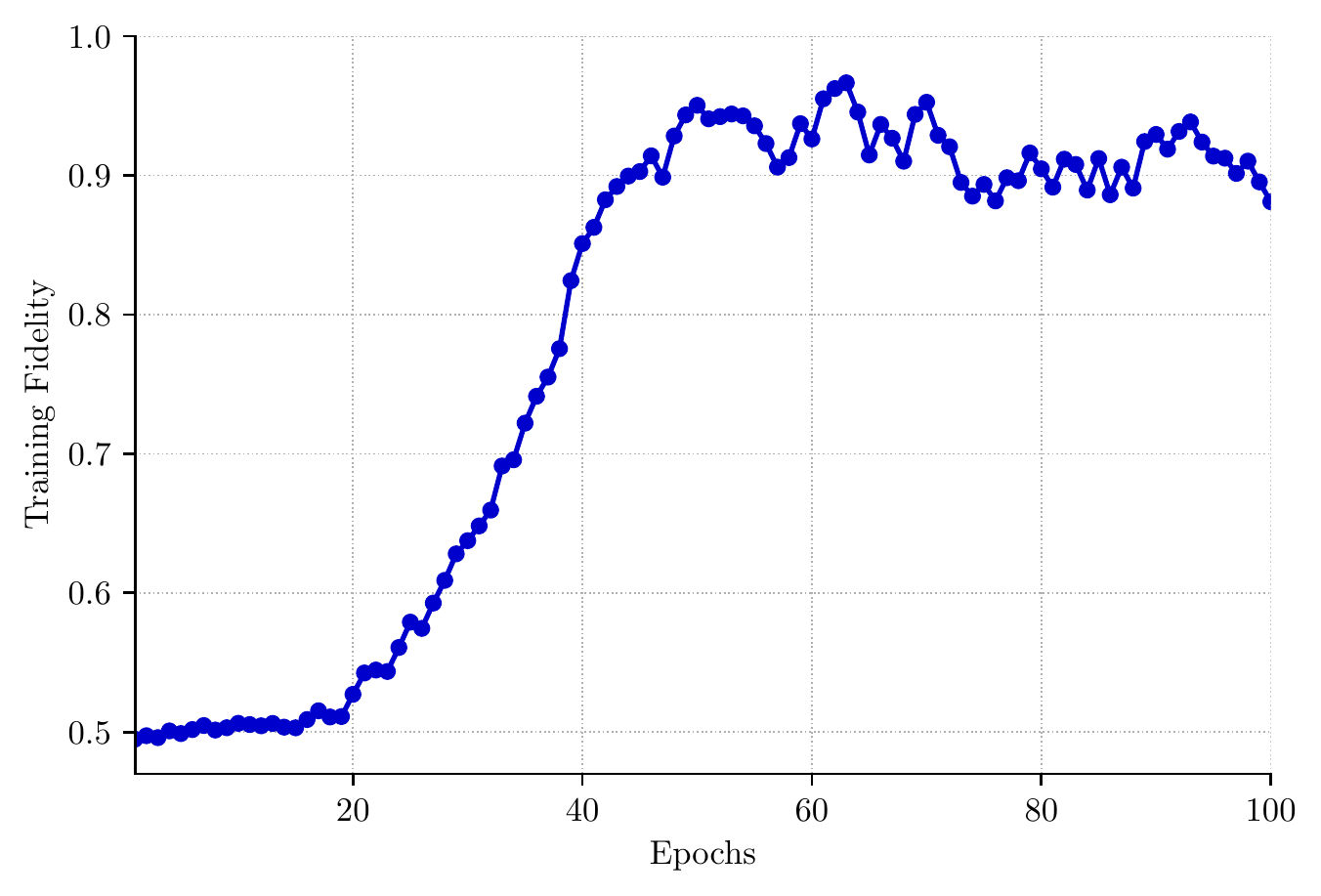}
  \caption{$p = 0.37$}
  \label{fig:d}
\end{subfigure}
\begin{subfigure}{.33\textwidth}
  \centering
  \includegraphics[width=0.9\textwidth]{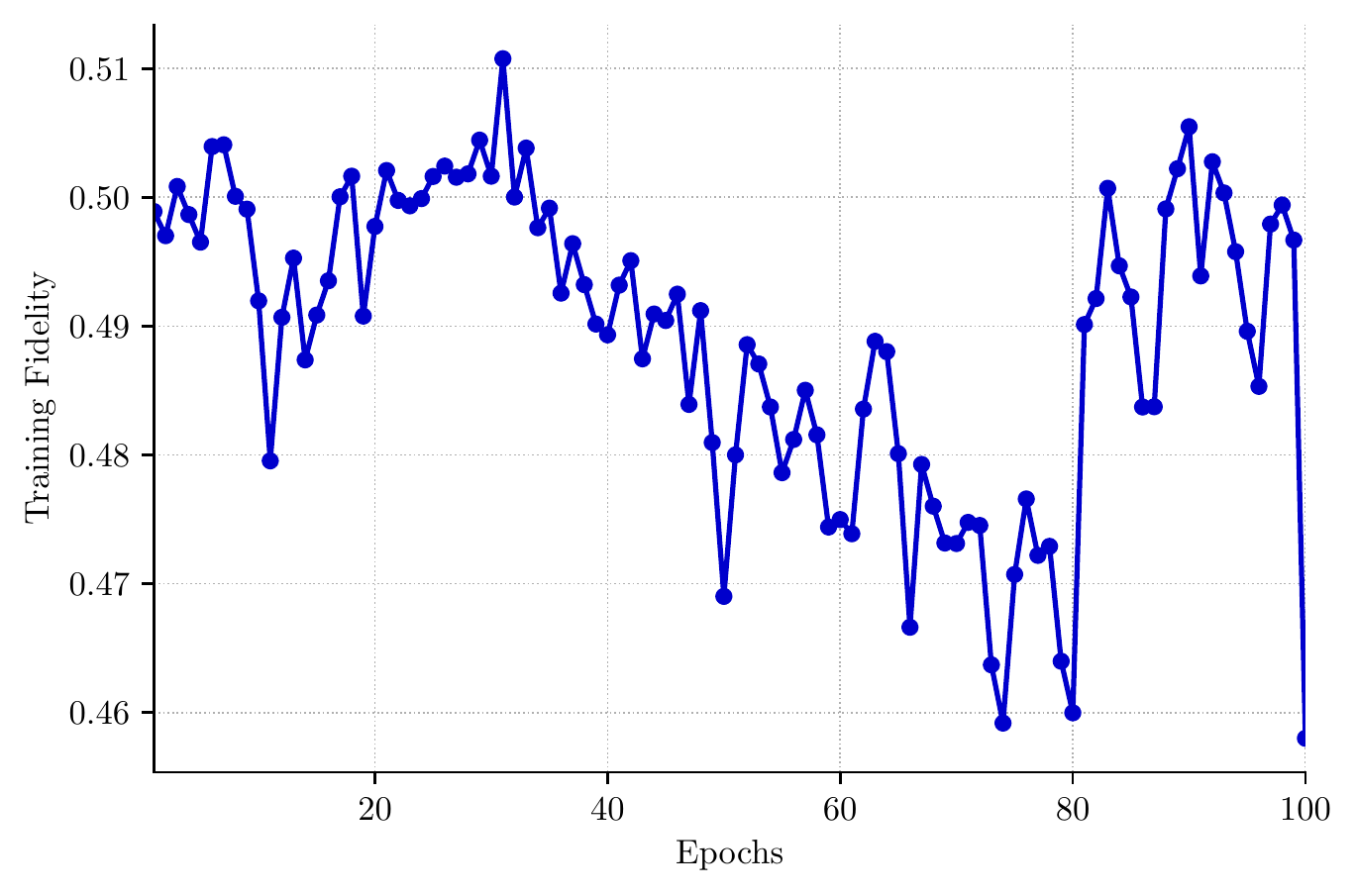}
  \caption{$p = 0.4$}
  \label{fig:e}
\end{subfigure}
\begin{subfigure}{.33\textwidth}
  \centering
  \includegraphics[width=0.9\textwidth]{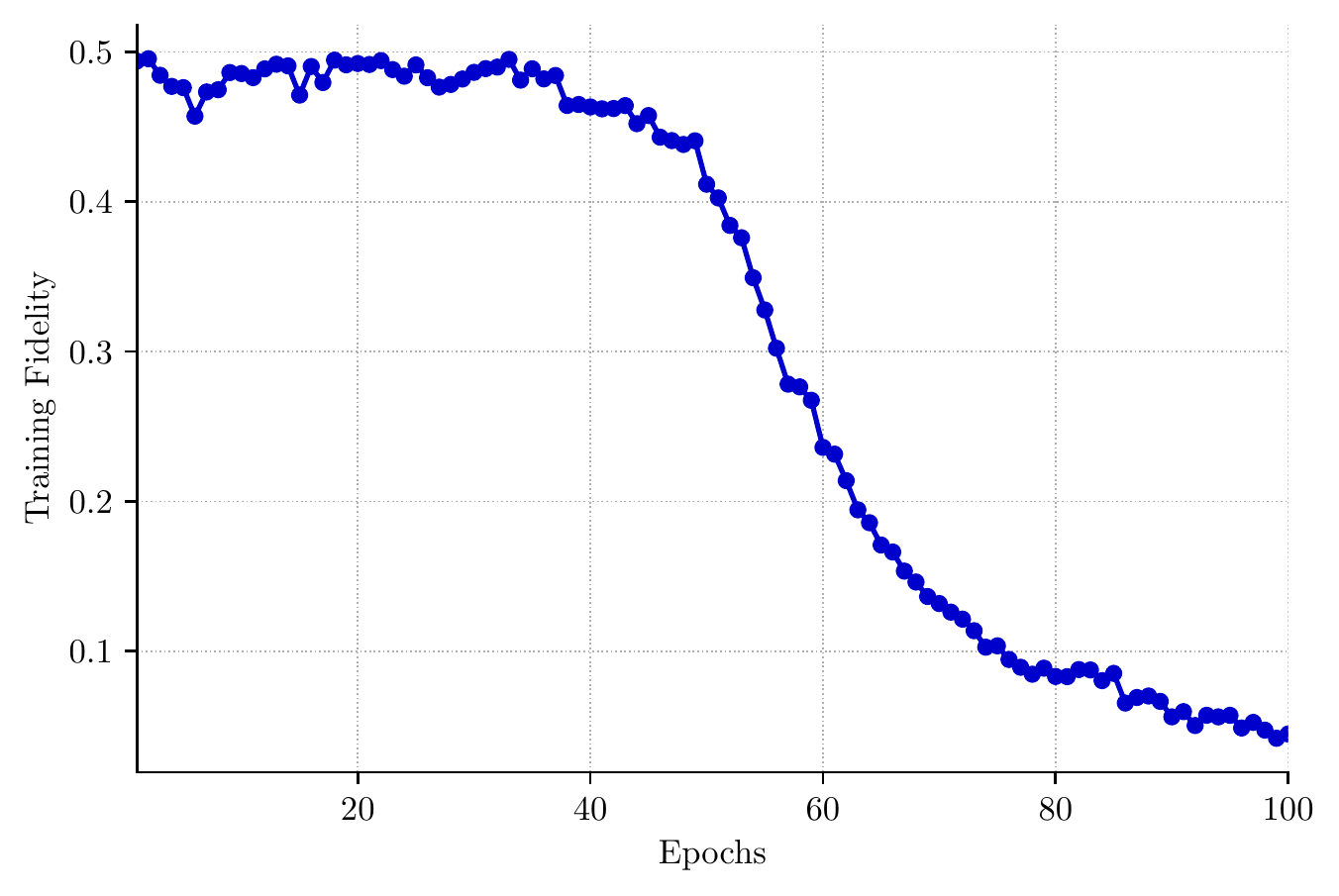}
  \caption{$p = 0.4$}
  \label{fig:f}
\end{subfigure}
\caption{Training fidelities of a [2,1,2] QAE trained with pairs affected by the bit-flip channel for various strengths. In all experiments, we trained on 100 pairs, using $\epsilon = 0.1$ and $\eta = 1/4$. We can indeed see that the larger $p$ is, the slower the convergence. For instance, one can compare the epoch at which the fidelity crosses the 0.9 threshold.}
\label{fig:all_train_bitflip}
\end{figure*}

Thus, our QAE almost perfectly reconstructs GHZ states affected by bit-flips of arbitrary probability.
\medskip

Since increasing the bit-flip probability does not affect our QAE's performance, one could wonder whether the strength of the bit-flip channel (i.e. the bit-flip probability) is totally irrelevant in the whole experiment. Actually, it is not the case. The strength of the channel plays a critical role in the training part : the stronger the noise, the slower the convergence. And if the noise is too strong, the QAE might never converge to the desired state. For instance, if $p$ is too large, we will have more $\frac{1}{\sqrt{2}}(\ket{01} + \ket{10})$ states than GHZ states, thus the QAE will erroneously conclude that this is the original state to learn. This is represented in Figure \ref{fig:all_train_bitflip}.

In sub-figures \ref{fig:a} to \ref{fig:d}, the QAE is able to achieve an almost-perfect fidelity. However, in sub-figures \ref{fig:e} and \ref{fig:f}, we observe two different behavior : the fidelity stagnates around 0.5 or goes to 0 instead of 1. This is because when $p = 0.4$, by Eq. \eqref{eq:proba_stateflip}, we have $52\%$ of GHZ state (and thus $48\%$ of the $\frac{1}{\sqrt{2}}(\ket{01} + \ket{10})$ state). Since these probabilities are very close to $50\%/50\%$, chances are we train our QAE on an almost-even number of pairs for each states, hence it will not be able to determine which state is the original one and which is the noisy one (sub-figure \ref{fig:e}). Furthermore, we may also train our QAE on more pairs of the non-GHZ state, hence it will be misled to deduce that the GHZ states are actually the noisy ones, and will convert them in the $\frac{1}{\sqrt{2}}(\ket{01} + \ket{10})$ state (sub-figure \ref{fig:f}).

\subsection{Quantum depolarizing channel}

Considering the ability of our QAE to denoise states corrupted by bit-flips, we further attempted to denoise states affected by the quantum depolarizing channel (QDC), which is a more general form of noise. For an input state with density matrix $\rho$ representing $m$ qubits, the QDC with noise strength $p$ is defined as
\begin{align}
    \rho \xrightarrow[]{QDC} (1-p) \rho + p  
    \big( \frac{1}{m} \cdot \large{\textbf{1}} \big),
\end{align}
where $\frac{1}{m} \cdot \large{\textbf{1}}$ is the density matrix of the maximally mixed state. One can show that this is equivalent to applying exactly one of the matrices $\{I,X,Y,Z\}$ (the Pauli matrices and the identity) to each qubit, with respective probabilities $\{1 - \frac{3p}{4}, \frac{p}{4}, \frac{p}{4},\frac{p}{4}\}$.
\medskip 

We trained a [2,1,2] QAE with 100 training pairs submitted to the QDC with noise strength $p=0.2$. The training fidelity is represented in Figure \ref{fig:train_depo}.

\begin{figure}[ht]
\centering
\includegraphics[scale = 0.5]{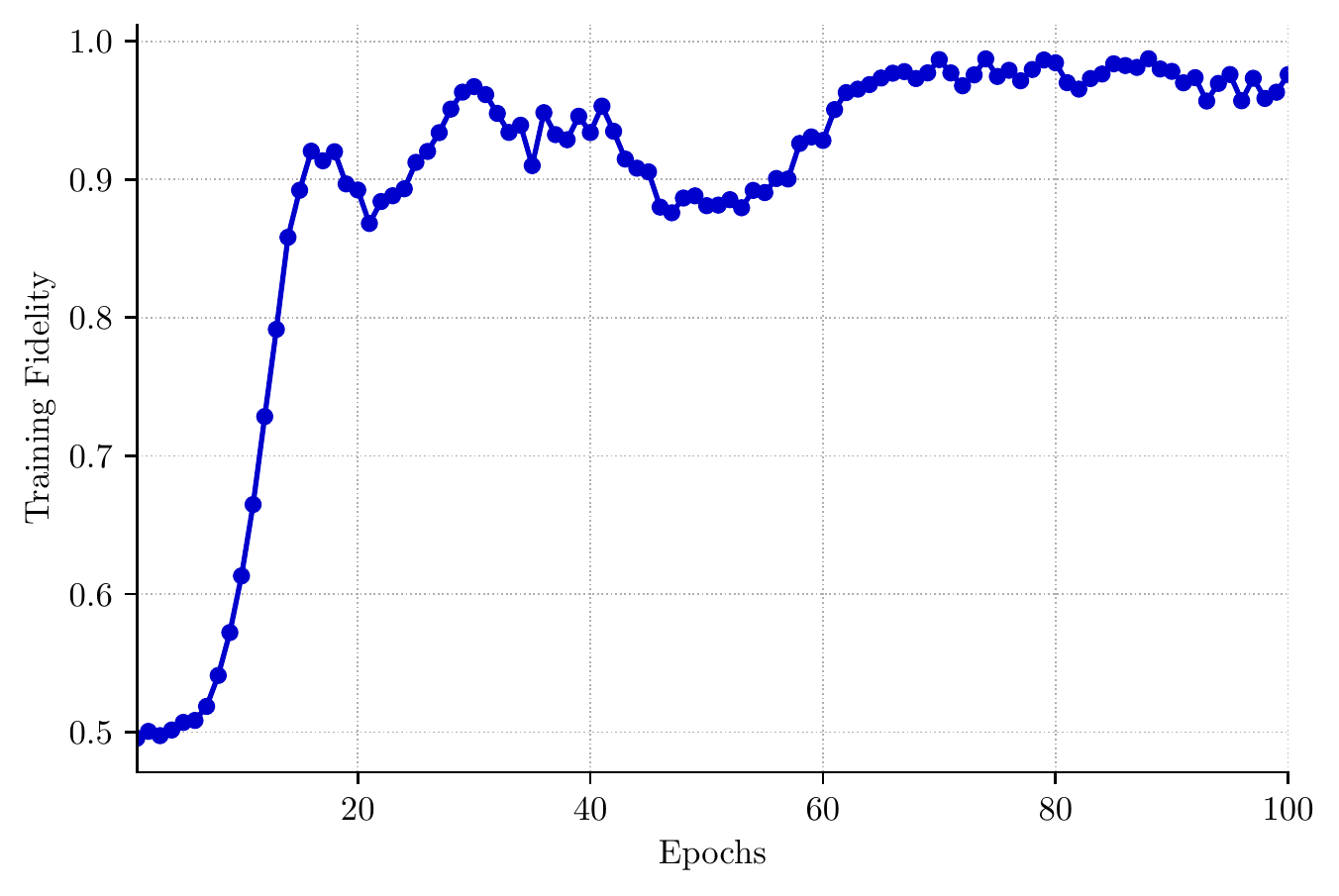}
\caption{Training fidelity of a $[2,1,2]$ QAE using 100 training pairs affected by the QDC with noise strength $p=0.2$. We used $\epsilon = 0.1$ and $\eta = 1/4$.}
\label{fig:train_depo}
\end{figure}

Once again, we tested the robustness of our model when states are corrupted with stronger noise. The results, displayed in Figure \ref{fig:fid_depo}, are impressive. As before, the theoretical fidelity can be computed as a function of the noise strength $p$. 
\begin{proposition}
For a $m$-qubit GHZ state corrupted by the QDC with noise strength $p$, the theoretical fidelity with GHZ is
\begin{align}
\sum\limits_{\substack{k=0 \\ k \text{ even}}}^m \binom{m}{k} \Big(\frac{p}{4} \Big)^k \Big(1 - \frac{3p}{4} \Big)^{m-k} + 2^{m-1} \cdot \Big(\frac{p}{4}\Big)^m.
\end{align}
\end{proposition}

\begin{proof}
Every state output by the QDC that is not GHZ (up to a global phase) will be orthogonal to it, hence their fidelity will be 0. And for the QDC to output the GHZ state (up to a global phase), there is only two possibilities : 
\begin{itemize}
    \item an even number of qubits were affected by $Z$ and the rest by $I$
    \item an even number of qubits were affected by $Y$ and the rest by $X$
\end{itemize}
Since the fidelity of these states with GHZ is 1, the theoretical fidelity is thus
\begin{align}
    &\sum\limits_{\substack{k=0 \\ k \text{ even}}}^m \binom{m}{k} \Big(\frac{p}{4} \Big)^k \Big(1 - \frac{3p}{4} \Big)^{m-k} \nonumber \\ +  & \sum\limits_{\substack{k=0 \\ k \text{ even}}}^m \binom{m}{k} \underbrace{\Big(\frac{p}{4} \Big)^k \Big(\frac{p}{4} \Big)^{m-k}}_{(\frac{p}{4})^m}.
\end{align}

Since $\sum\limits_{\substack{k=0 \\ k \text{ even}}}^m \binom{m}{k} = 2^{m-1}$, we get the final result.
\end{proof}

\begin{figure}[ht]
\centering
\includegraphics[scale = 0.5]{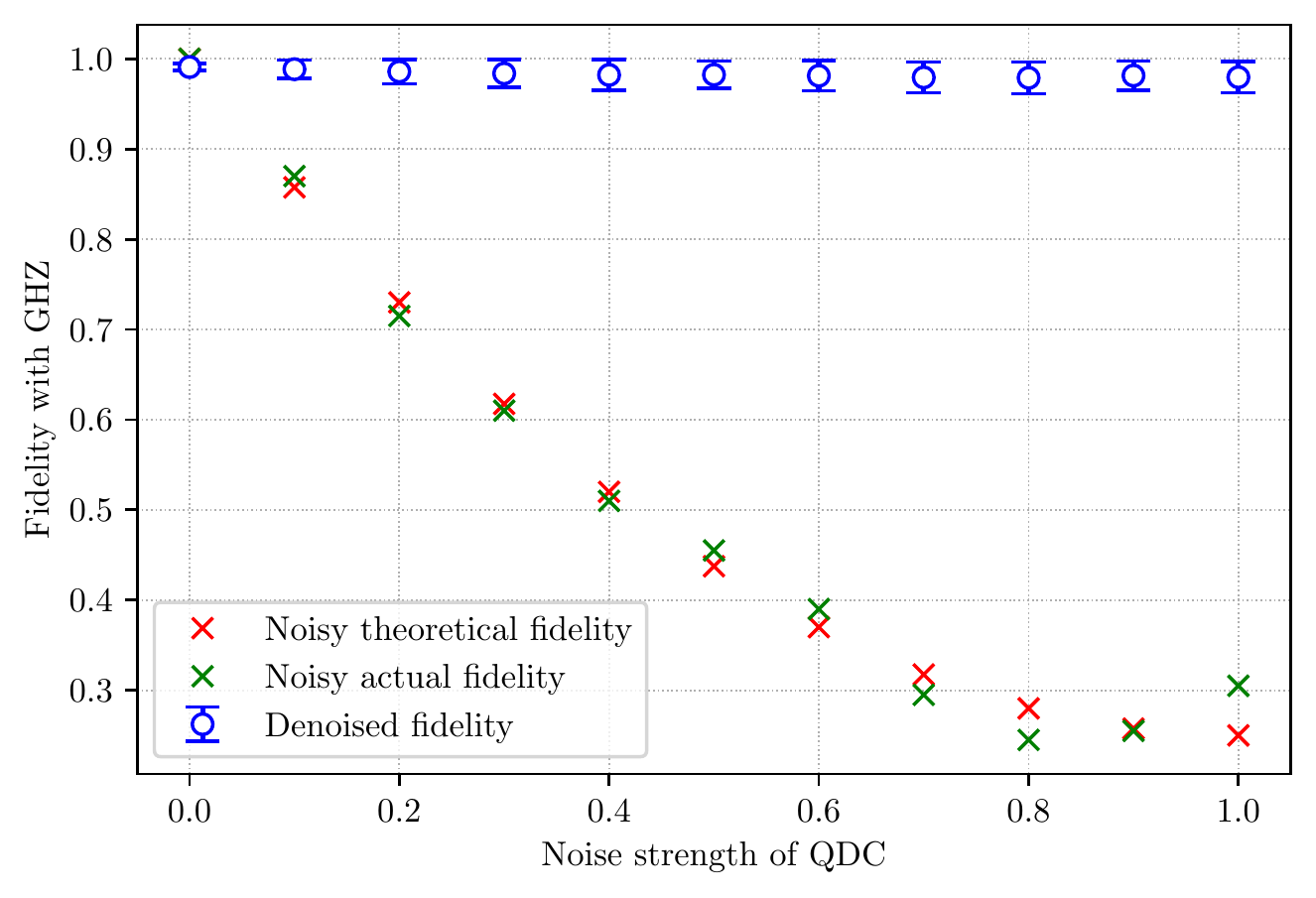}
\caption{Mean and deviation of the fidelity of the $[2,1,2]$ network, trained with $p = 0.2$, over different noise strength of the QDC. For each $p$, 200 test states were drawn. The green points are the fidelities of the test states. The red points are their theoretical fidelities.}
\label{fig:fid_depo}
\end{figure}

As it was the case with the bit-flip channel, our QAE is indifferent to the strength of the QDC and displays almost perfect reconstruction of GHZ states in every case.
\medskip

The QAE also features impressive generative abilities. Since the QAE is able to effectively recreate GHZ states independently of the noise strength, it should be able to recreate them without taking any input (i.e. starting with the state $\ket{0\cdots0}$), and thus act as a generative model for non-noisy GHZ states. We ran 200 simulations, and found that the average fidelity between the states created by our QAE and the GHZ state is 0.96 with a standard deviation of 0.01, thus supporting the above claim.
\medskip

In a second experiment, we trained a $[3,1,3]$ QAE on the same noise distribution. The training fidelity is represented in Figure \ref{fig:train_depo3}.

\begin{figure}[ht]
\centering
\includegraphics[scale = 0.5]{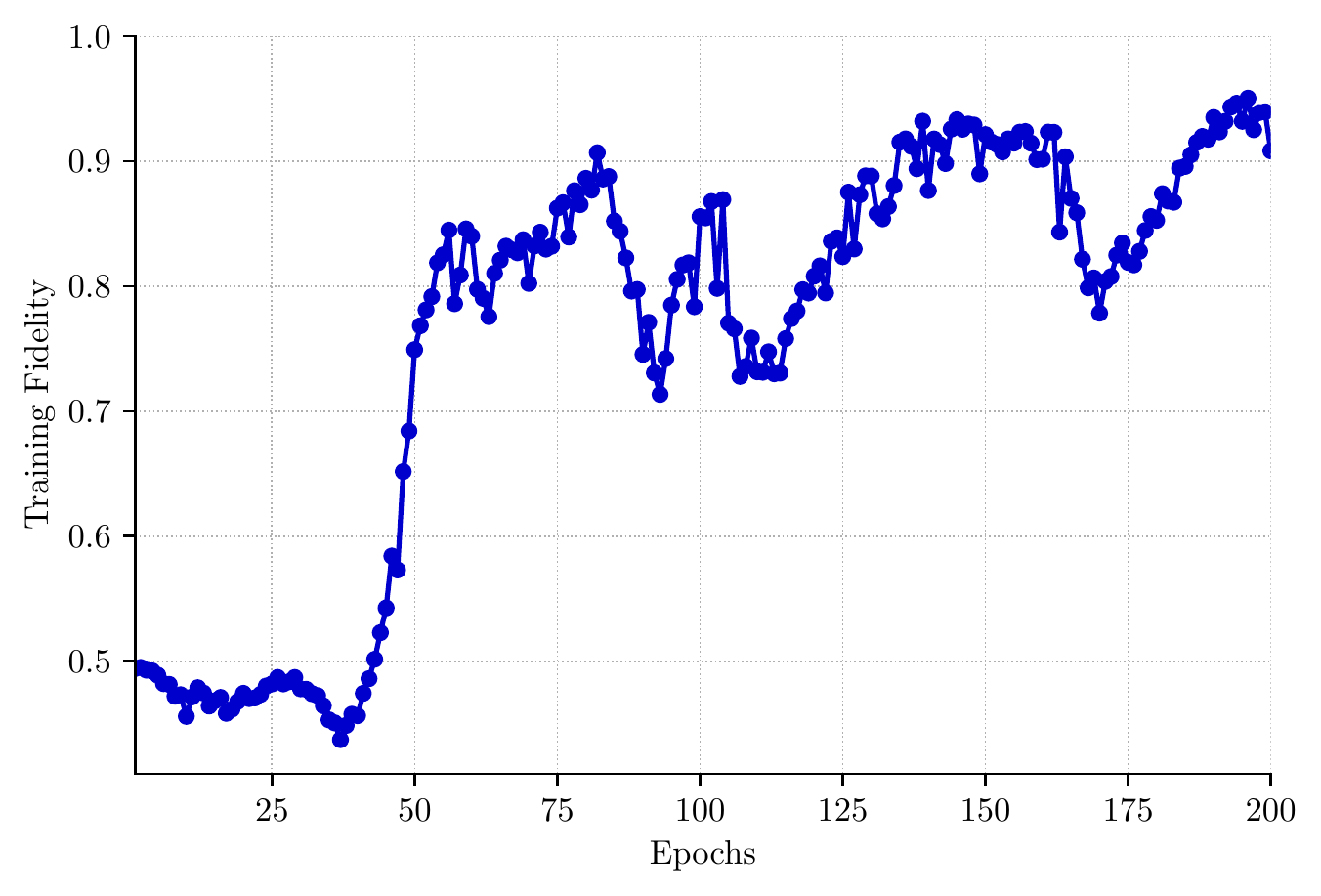}
\caption{Training fidelity of a $[3,1,3]$ QAE using 150 training pairs affected by the QDC with noise strength $p=0.2$. We used $\epsilon = 0.1$ and $\eta = 1/4$.}
\label{fig:train_depo3}
\end{figure}

The robustness of the $[3,1,3]$ QAE is displayed in Figure \ref{fig:fid_depo31}. Once again, the QAE is rather insensitive to stronger noise. However, the fidelity is not as high as with the $[2,1,2]$ QAE (between 0.9 and 0.95 here), and the standard deviation is quite large, especially for very strong levels of noise. 

\begin{figure}[ht]
\centering
\includegraphics[scale = 0.5]{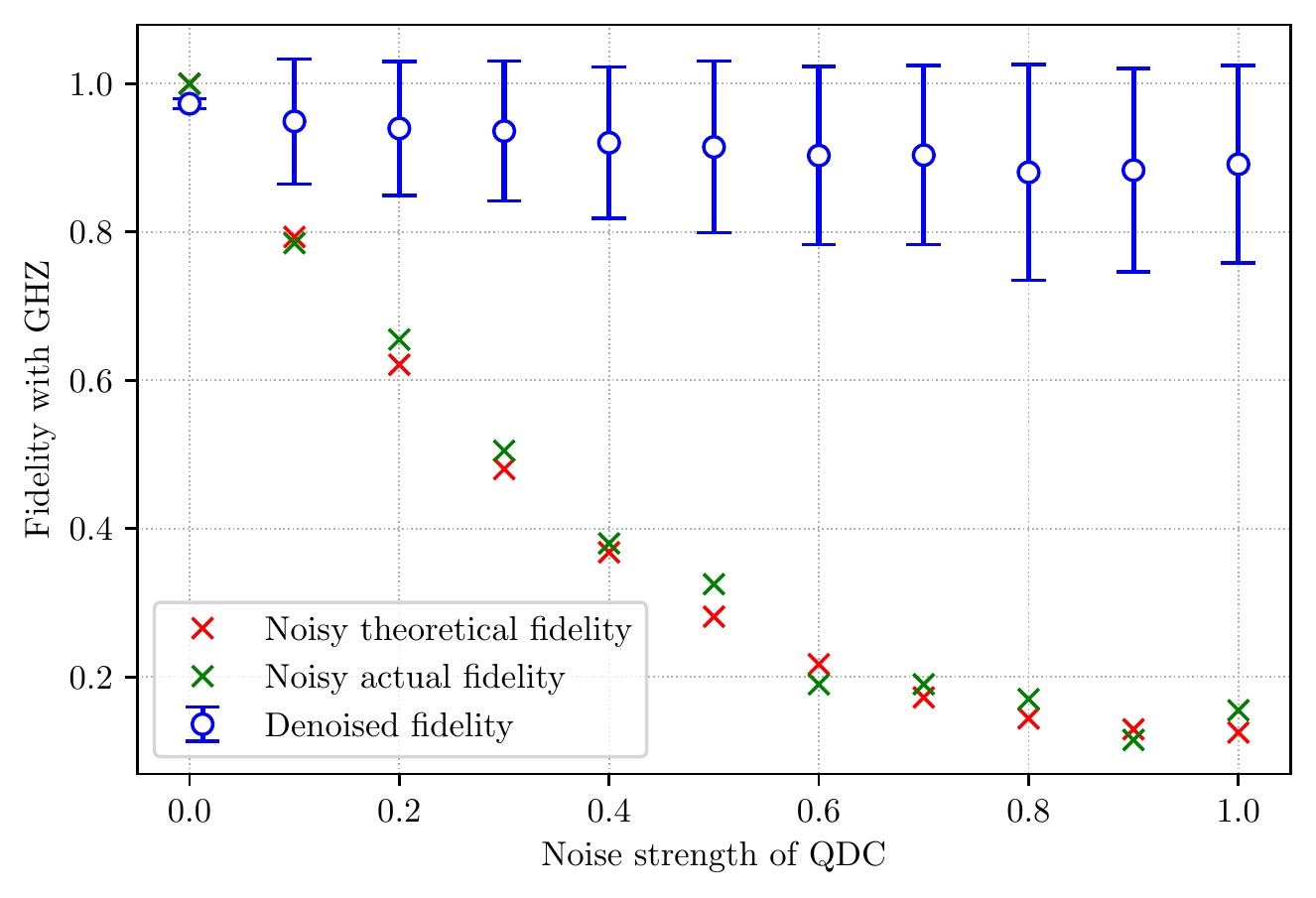}
\caption{Mean and deviation of the fidelity of the $[3,1,3]$ network, trained with $p = 0.2$, over different noise strength of the QDC. For each $p$, 200 test states were drawn. The green points are the fidelities of the test states. The red points are their theoretical fidelities.}
\label{fig:fid_depo31}
\end{figure}

To further improve the fidelity and dampen the standard deviation, we can apply our $[3,1,3]$ QAE twice. The results are presented in Figure \ref{fig:fid_depo32}.

\begin{figure}[ht]
\centering
\includegraphics[scale = 0.5]{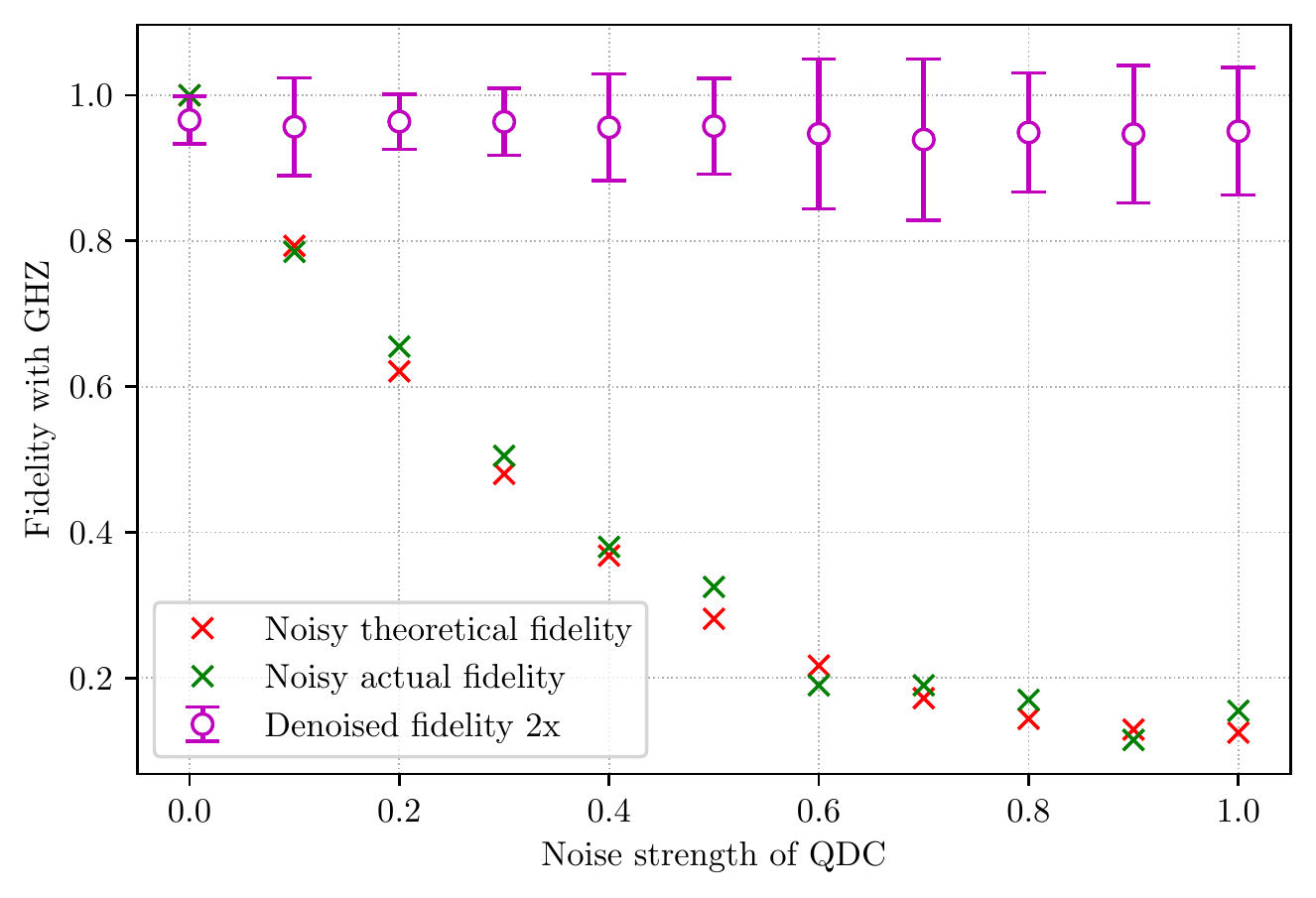}
\caption{Mean and deviation of the fidelity of the $[3,1,3]$ network, trained with $p = 0.2$ and applied twice, over different noise strength of the QDC. For each $p$, 200 test states were drawn. The green points are the fidelities of the test states. The red points are their theoretical fidelities.}
\label{fig:fid_depo32}
\end{figure}

Our $[3,1,3]$ QAE can also be used as a generative model. Running the same experiment as with the $[2,1,2]$ QAE, we found that the states created by our QAE achieved an average fidelity of $0.95 \pm 0.01$ with the 3-qubit GHZ state.
\medskip

To conclude this section, we can visualize the effect of our $[2,1,2]$ QAE on noisy states, using the 'States City' plot feature of Qiskit (see Figure \ref{fig:states_city}). It can be observed how the QAE is re-heightening the relevant states, i.e. the states corresponding to the density matrix of the GHZ state.

\begin{figure*}
\begin{subfigure}{.5\textwidth}
  \centering
  \includegraphics[width=0.9\textwidth]{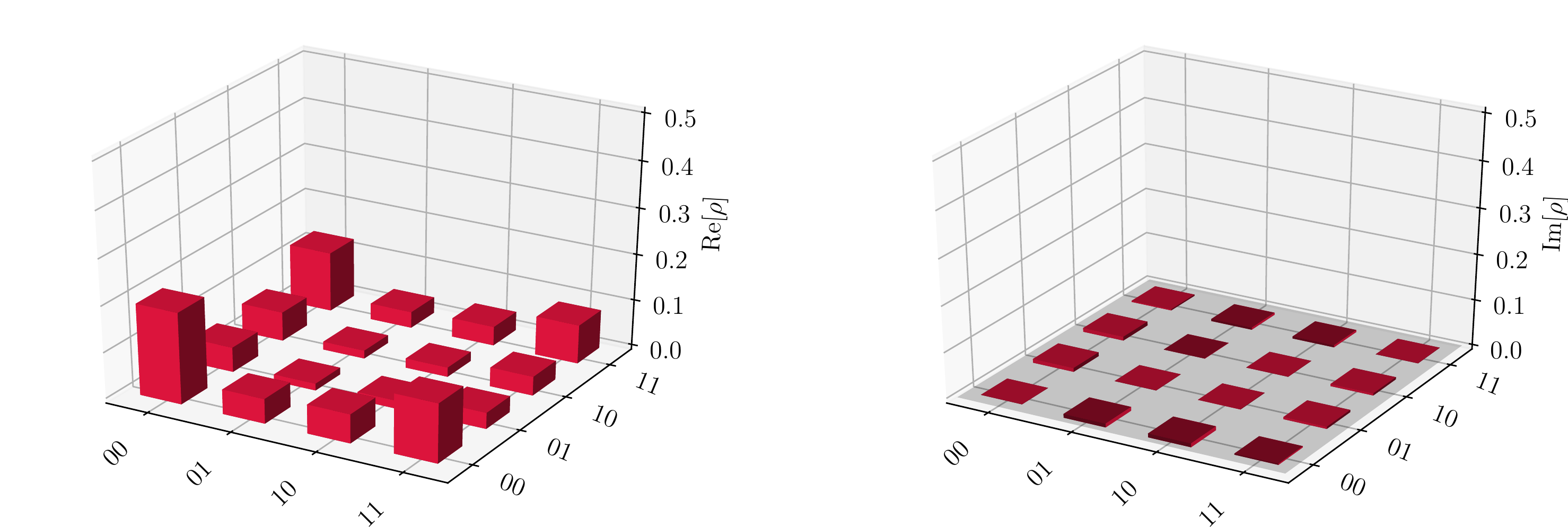}
  \caption{Noisy states city}
\end{subfigure}
\begin{subfigure}{.5\textwidth}
  \centering
  \includegraphics[width=0.9\textwidth]{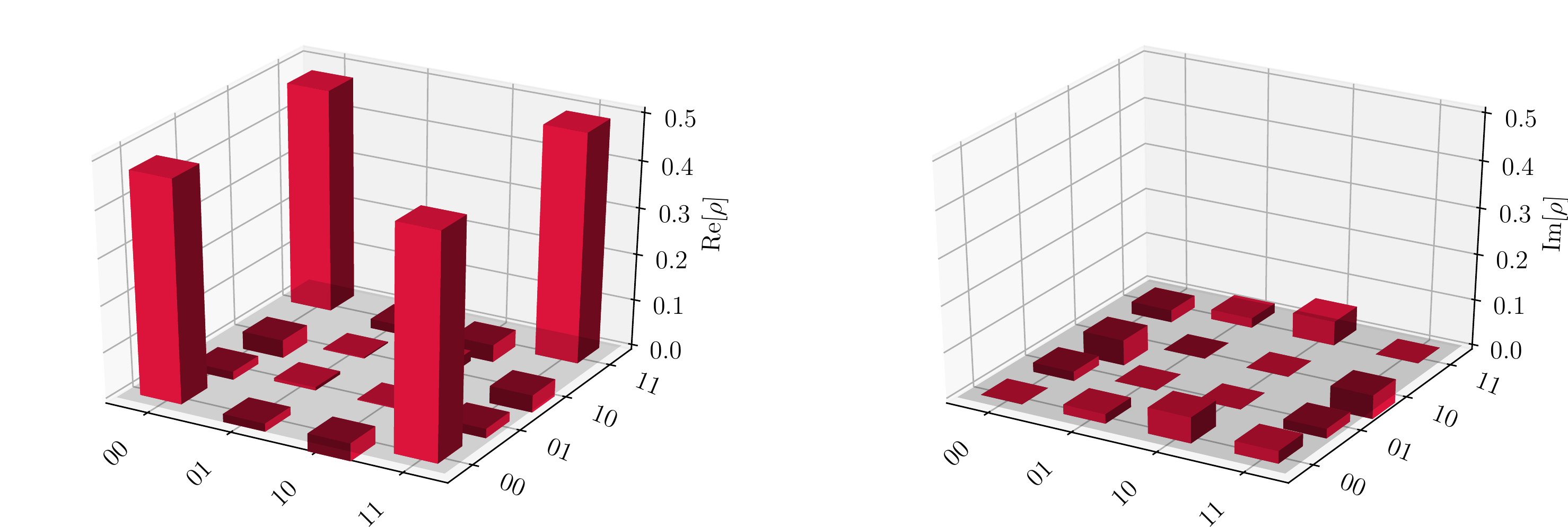}
  \caption{Denoised states city}
\end{subfigure}
\caption{States city of noisy and denoised states. We used 200 GHZ states affected by the QDC with noise strength $p=0.4$. The states city are respectively the average initial state vector and the average density matrix output by the QAE.}
\label{fig:states_city}
\end{figure*}

\subsection{Noise Robustness}

It is interesting to assess the robustness of our QAEs to internal noise, as it is ubiquitous in a real quantum device. To do so, we introduce a Gaussian noise in the unitary gates composing our QAEs. Using Eq. \eqref{eq:unitaries}, the addition of a noisy term $\delta$ leads to the following equation for the noisy gates :

\begin{align}
    \widetilde{U} = e^{i(K + \delta)} = e^{i\delta} \cdot U
\end{align}

For a meaningful comparison, the introduced noise $\delta$ has to be significantly smaller than $K$, which was found in previous simulations to be about $-0.3 \pm 0.25$. We choose a Gaussian noise $\delta$ centered on $0$, and tested the robustness of the QAE while increasing its standard deviation. The results are illustrated in Figures \ref{fig:noise_rob2} and \ref{fig:noise_rob3}, respectively for the [2,1,2] and [3,1,3] QAEs. In each case, 200 states are picked from the noisy QDC with noise strength $p = 0.3$. The results show a significant robustness to noise, as the QAEs still increase the fidelity of the denoised states, until the standard deviation of the gates' noise reaches roughly one third of the average value of the coefficients $K$. Furthermore, Figures \ref{fig:noise_rob2} and \ref{fig:noise_rob3} display that the GHZ states are almost perfectly reconstructed by the QAEs as long as the standard deviation remains below 0.03.

\begin{figure}[ht]
\centering
\includegraphics[scale = 0.5]{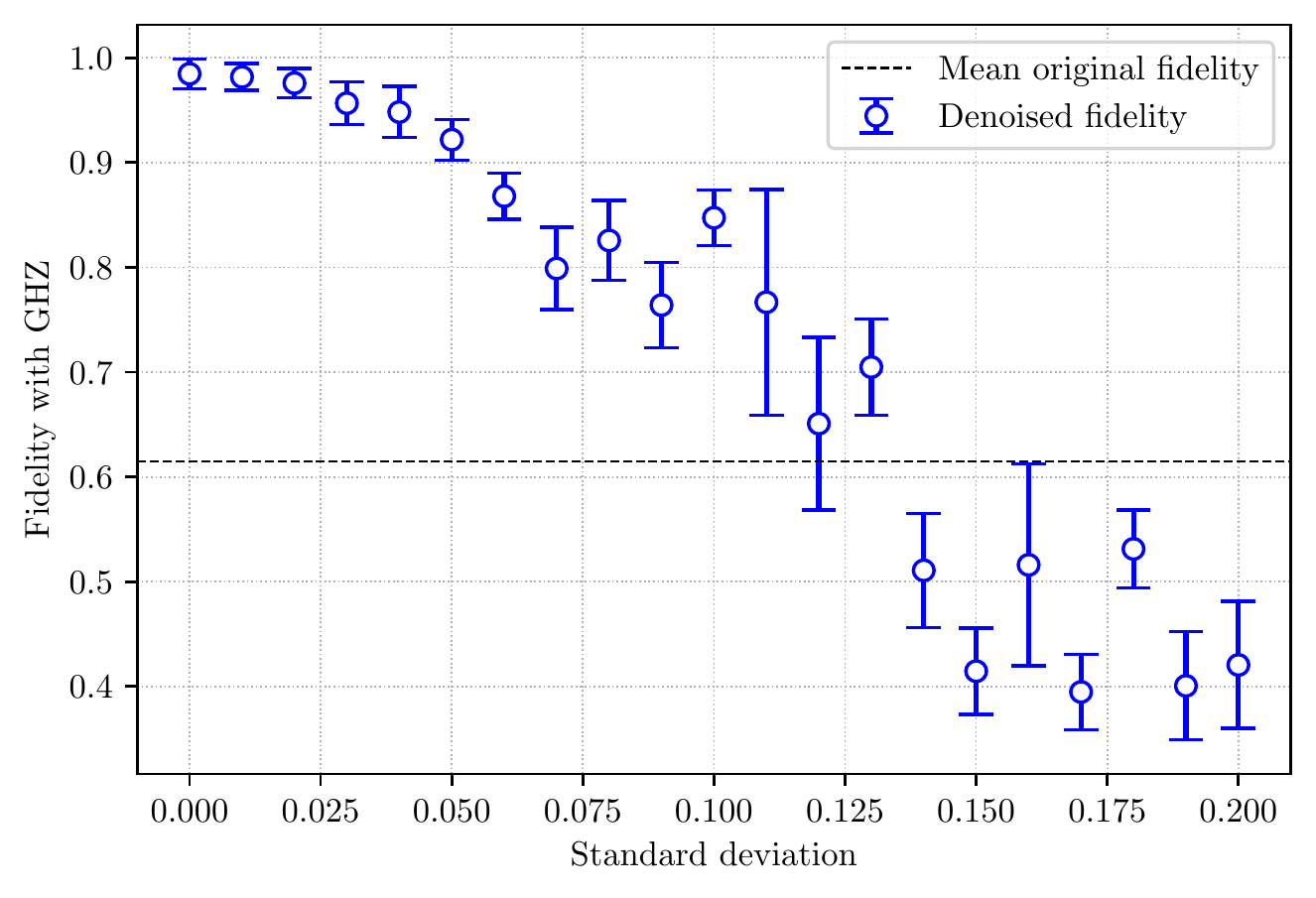}
\caption{Noise robustness of the [2,1,2] QAE}
\label{fig:noise_rob2}
\end{figure}

\begin{figure}[ht]
\centering
\includegraphics[scale = 0.5]{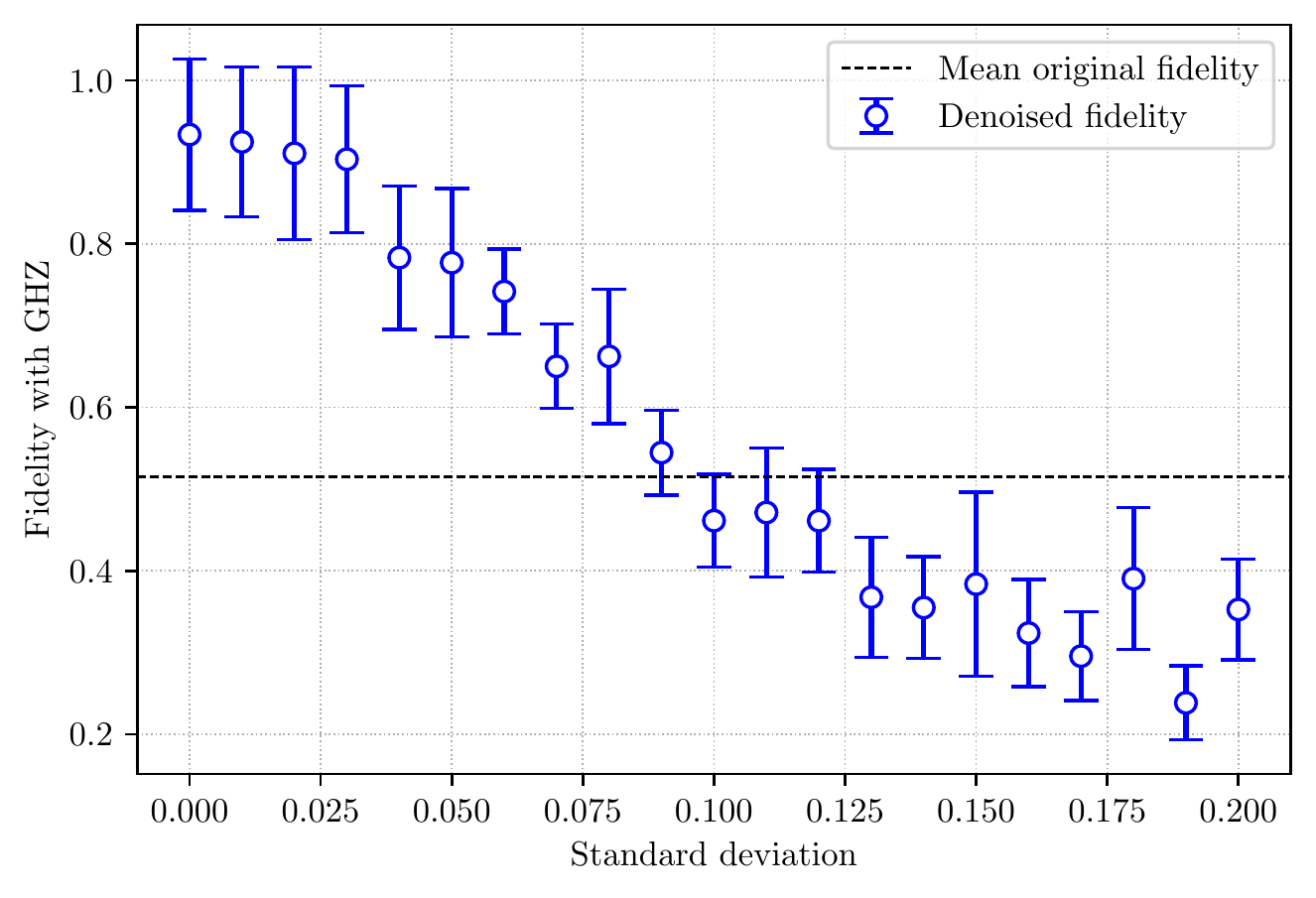}
\caption{Noise robustness of the [3,1,3] QAE}
\label{fig:noise_rob3}
\end{figure}

\subsection{Findings}

In this section, we have demonstrated that the QAEs studied in this paper are able to (almost) perfectly denoise GHZ states, and at the same time generate them without explicitly being given the goal state. Furthermore, we have shown that these QAEs are extremely robust to internal noise.

We emphasize once again that during the whole training phase, the QAE only sees noisy GHZ states, both on the input and on the target. It is then able to perceive the underlying structure behind all these noisy states, and to recompose it.

\section{Application : Quantum Secret Sharing} \label{QSS}

We detailed the creation and training of QAEs. We now discuss a practical application of them in the context of QSS. We analyze the limitation noise confers to QSS and show how QAEs can save QSS protocols by dampening the impact of noise. But first, we explain the basis of a QSS protocol.

\subsection{Preliminaries}

A Secret Sharing protocol is a method that allows a secret to be split between participants, so that each share does not give any information about the secret, but if combined together, they allow the participants to recover the secret. 

In a classical setting, for instance, a secret could be a two digit number (i.e. between $00$ and $99$). If we have two participants, one possible scheme would be to give one digit of the number to each participant. However, this scheme is not considered as a secure secret sharing protocol, as having a digit of a number gives considerable information about that number. In contrast to this insecure protocol, consider a scheme where a randomly drawn number is added to the secret and the sum is reduced modulo 100. We then give the first participant the result, and the second party the drawn number. In this example, none of the participants can learn anything about the secret alone. But by collaborating, they can fully retrieve it.

Formally, a secured Secret Sharing protocol must ensure that :
\begin{itemize}
    \item collaboration allows the participants to recover the secret
    \item shares alone do not allow any participant to reduce the search state space in which the secret lies
\end{itemize}

Quantum Physics helps building secure Secret Sharing protocols by distributing entangled states between participants. We refer to QSS as described in \cite{hillery1998quantum}. The interest of QSS is both to make it difficult for a malignant participant to cheat and for an eavesdropper to intercept information, as they will be (with a high probability) detected by the other participants.
\medskip

We consider here a GHZ triplet, shared between Alice, Bob and Charlie. Initially, Charlie possesses the whole triplet. He sends to Alice and Bob the first and second qubits, and then each of the parties makes a measurement in either the $\mathcal{X}$ or $\mathcal{Y}$ basis (randomly selected), defined by
\begin{align}
    \mathcal{X} &: \left\{ \begin{array}{c}
   \ket{+x} = \frac{1}{\sqrt{2}}(\ket{0} + \ket{1}) \\
   \ket{-x} = \frac{1}{\sqrt{2}}(\ket{0} - \ket{1})
\end{array} \right.,
\\ \mathcal{Y} &: \left\{ \begin{array}{c}
   \ket{+y} = \frac{1}{\sqrt{2}}(\ket{0} + i\ket{1}) \\
   \ket{-y} = \frac{1}{\sqrt{2}}(\ket{0} - i\ket{1})
\end{array} \right..
\end{align}

Once all measurements are made, each participant announces the basis he has chosen. As we will show, in half of the cases Alice and Bob will know the result of Charlie's measurement, provided that they collaborate (and they can not find it if they do not). They will then have a secret bit in common. By repeating this protocol multiple times, Charlie will be able to establish a secret key with Alice and Bob, that can be used, for instance, for sending a classical message encrypted with the shared key (like a One-Time-Pad \cite{miller1882telegraphic}).
\medskip

To show how Alice and Bob can know the result of Charlie's measurement, we can re-write the GHZ triplet in all the possible combinations of the $\mathcal{X}$ and $\mathcal{Y}$ basis. We have :

\begin{align}
    \ket{\psi} &= \frac{1}{\sqrt{2}} (\ket{000} + \ket{111}) \\
    &=\medmath{\frac{1}{2} \Big[ \big(\ket{+x}\ket{+x} + \ket{-x}\ket{-x} + \ket{+y}\ket{-y} + \ket{-y}\ket{+y} \big)\ket{+x}} \nonumber \\
    & ~~~~~ \medmath{+ \big(\ket{+x}\ket{-x} + \ket{-x}\ket{+x} + \ket{+y}\ket{+y} + \ket{-y}\ket{-y}\big)\ket{-x}} \nonumber \\
    &~~~~~\medmath{+ \big(\ket{+x}\ket{-y}+ \ket{-x}\ket{+y} + \ket{+y}\ket{-x} + \ket{-y}\ket{+x} \big)\ket{+y}} \nonumber \\
    &~~~~~\medmath{+ \big(\ket{+x}\ket{+y} + \ket{-x}\ket{-y} + \ket{+y}\ket{+x} + \ket{-y}\ket{-x} \big)\ket{-y} \Big]}.
\label{eq:qss}
\end{align}

Hence, we can see that the 'valid' bases (under which Charlie will be able to have a shared bit with Alice and Bob) are $\{\mathcal{XXX}, \mathcal{YYX}, \mathcal{XYY}, \mathcal{YXY} \}$ or simply when an even number of participants takes a measurement according to the $\mathcal{Y}$ basis (as measuring a basis in the other one leads to a probability of $1/2$ for each state of the basis). When the three participants have chosen one of the above bases, by sharing their results, Alice and Bob can know the result of Charlie, but they can't know it individually.

\subsection{Noisy QSS Protocol}

\begin{figure*}
\centering
\includegraphics[width = \textwidth]{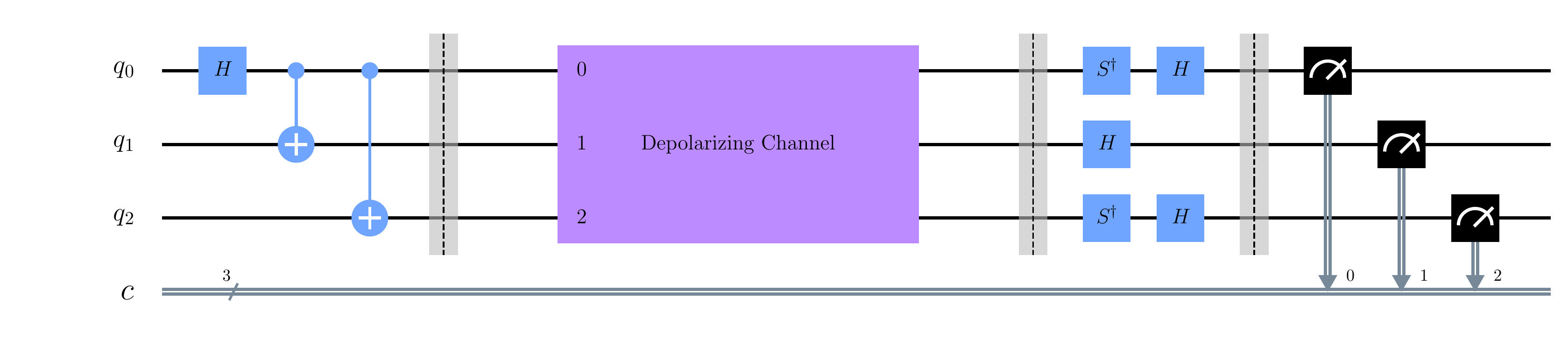}
\caption{Noisy QSS circuit. Here, the basis chosen is $\mathcal{YXY}$.}
\label{fig:QSS}
\end{figure*}

Now let us assume that participants do not have access to a noise-free GHZ triplet (for instance because of hardware constraints or shortcomings), or that the one they have has been corrupted by noise. In this section, we analyze how that will impact the protocol.
\medskip

To do so, we suppose that the GHZ triplet is corrupted by the QDC with strength $p$. This means that each of the qubits will be affected by a particular Pauli matrix with probability $\frac{p}{4}$ (for each of the three matrices), and will stay unharmed with probability $1 - \frac{3p}{4}$. Hence, to quantify the impact of the QDC, we only have to assess the effect of Pauli matrices on the basis $\mathcal{X}$ and $\mathcal{Y}$. We have 
\begin{align}
\left\{ \begin{array}{lcl}
   X \ket{\pm x} = \pm \ket{\pm x}&  &  X \ket{\pm y} = \pm i \ket{\mp y}\\
   Y \ket{\pm x} = \mp i \ket{\mp x}&  &  Y \ket{\pm y} = \pm \ket{\pm y} \\
   Z \ket{\pm x} = \ket{\mp x}&  &  Z \ket{\pm y} = \ket{\mp y}
\end{array} \right.
\end{align}
i.e. $X$ affects the $\mathcal{Y}$ basis, $Y$ affects the $\mathcal{X}$ basis, and $Z$ affects both. By 'affecting' the basis, we mean that it inverts the basis' states, which will disorder Eq. \eqref{eq:qss}. The coefficients are irrelevant here.

It then becomes straightforward to understand how the QDC affects the protocol. Let's suppose that we are in the case of a 'valid' basis, hence the result of this measurement is kept a part of the shared key. Let's further suppose that in the QDC, Alice's qubit is affected by $X$, and Bob and Charlie's qubits are lucky enough to stay still. Will the protocol fail ?

If Alice chooses to do her measurement according to the $\mathcal{X}$ basis, nothing will change, as $X$ leaves the $\mathcal{X}$ basis unharmed, so the protocol will succeed. However, if she chooses the $\mathcal{Y}$ basis, then the bit recovered by Alice and Bob (which is supposed to be Charlie's bit) will always be the wrong one. Indeed, let us assume for instance that Bob chooses the $\mathcal{Y}$ basis too and finds $\ket{-y}$, and Charlie chooses the $\mathcal{X}$ basis. If Alice finds $\ket{+y}$, then, by referring to Eq. \eqref{eq:qss}, Bob and her will assume that Charlie found $\ket{+x}$, when he actually found $\ket{-x}$, since Alice's measurement should have been $\ket{-y}$, had her qubit not been affected by $X$. Hence, in this case, the probability of getting the wrong bit (and thus, having the protocol failed) is 1/2. We can conduct the same reasoning for every other options (syndromes) in the QDC, which leads to the following result.

\begin{proposition}
If the QSS protocol is affected by the QDC with strength $p$, the probability $\Gamma$ of the protocol failing is
\begin{align}
    \Gamma = \frac{p}{2}(p^2 - 3p + 3).
\end{align}
\end{proposition}

\begin{proof}
Assuming that Charlie's qubit is left unharmed, we get for Alice and Bob :
\begin{itemize}
    \item $XI, YI, IX, IY, XX, YY, XY, YX, XZ, ZX$, $YZ, ZY \Longrightarrow$ the protocol fails half of the time
    \item $IZ, ZI \Longrightarrow$ the protocol always fails
    \item $II, ZZ \Longrightarrow$ the protocol never fails (we can see in Eq. \eqref{eq:qss} that if both signs of the two first qubits change, that will leave the protocol unharmed)
\end{itemize}

We can then obtain the probability $\Gamma_I$ of the protocol failing when Charlie's qubit is affected by $I$ (i.e. is unharmed) by multiplying the probabilities of failing and the probabilities of having each syndrome in the QDC :

\begin{align}
    \Gamma_I &= \frac{1}{2}\Big[ 4 \cdot \frac{p}{4}\big(1 - \frac{3p}{4} \big) + 8 \cdot \big(\frac{p}{4} \big)^2 \Big] + 2 \cdot\frac{p}{4} \big(1 - \frac{3p}{4} \big)\nonumber \\
    &= \frac{1}{2}\Big[p - \frac{3p^2}{4} + \frac{p^2}{2} \Big] + \frac{p}{2} - \frac{3p^2}{8} \nonumber \\
    &= p(1 - \frac{p}{2}).
\end{align}

We can apply the same method when Charlie's qubit is affected by $X$ or $Y$, which yields :

\begin{align}
    \Gamma_X = \Gamma_Y = \frac{1}{2}.
\end{align}

Finally, when Charlie's qubit is affected by $Z$, we have :

\begin{align}
    \Gamma_Z = 1 - p + \frac{p^2}{2}.
\end{align}

Putting everything together, we get the total probability $\Gamma$ of the protocol failing : 

\begin{align}
    \Gamma &= (1 - \frac{3p}{4}) \cdot \Gamma_I + \frac{p}{4} \cdot \Gamma_X + \frac{p}{4} \cdot \Gamma_Y + \frac{p}{4} \cdot \Gamma_Z \nonumber \\
    &= \frac{p}{2}(p^2 - 3p + 3).
\end{align}
\end{proof}

We simulated the QSS protocol on Qiskit, as can be seen in Figure \ref{fig:QSS} (making a measurement in the $\mathcal{X}$ basis is equivalent to applying $H$ and making a measurement in the $ \mathcal{Z} = \{\ket{0}, \ket{1}\}$ basis, and similarly for the $\mathcal{Y}$ basis by adding a $S^\dagger$ before $H$). 
\medskip

In every simulation, the bases are randomly chosen, and if they correspond to a valid configuration, the secret bit is kept for the key. We made $p$ vary between 0 and 1, ran 1000 measurements for each value of $p$, then computed the number of differences between the two keys, weighted by the key length. The results are shown in Figure \ref{fig:QSS_noisy}.

\begin{figure}[ht]
\centering
\includegraphics[width = 0.49 \textwidth]{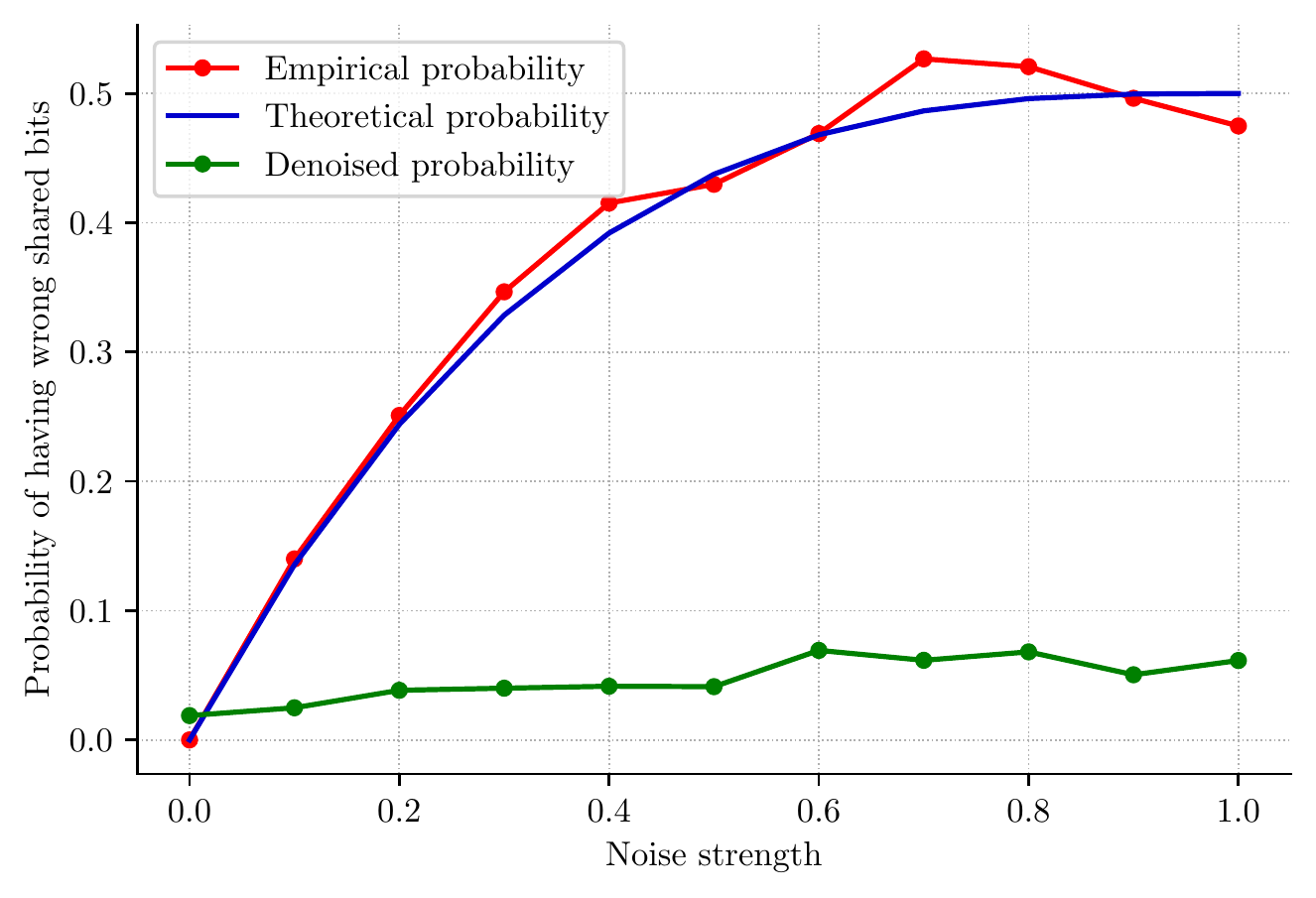}
\caption{Probabilities of the QSS protocol failing when submitted to the QDC with noise strength $p$. The blue line corresponds to $\Gamma$.}
\label{fig:QSS_noisy}
\end{figure}

The probability of failure increases with $p$, and at the limit $p=1$ (for a full-noise QDC), we have $\Gamma = \frac{1}{2}$, thus the protocol is not better than guessing a bit at random. Even a small value of $p$, like $p = 0.2$, will lead to $\Gamma = 0.24$, so on average approximately 1 bit out of 4 of the final secret key will be wrong, rendering it rather useless. 

Moreover, Charlie will not be able to detect if Alice or Bob is cheating, or if a malignant third party is trying to get a hold of the key, at it was the case in a non-noisy QSS protocol. Indeed, \cite{hillery1998quantum} showed that if Alice or Bob is cheating by getting both qubits and sending the other a qubit she/he carefully prepared, or if a third party is grafting an ancilla to the GHZ state to get any information, then it will irreducibly introduce errors. In a non-noisy QSS scheme, the participants can detect that, by revealing a fraction of the key, and see if there are any discrepancies. However, in a noisy QSS scheme, it will be challenging to determine if the error is due to a cheating act or simply due to the noisy channel. Although, cheating activities can still be detected, provided that the channel is not too noisy, and at the expense of a great number of simulations of the protocol, which may not be possible in practice. Indeed, the probability of the protocol failing in the presence of a cheater is 1/4. Hence, if the noise strength of the channel leads to a fewer probability of errors (for instance, when it is fewer than $0.2$ in the case of the QDC), then the participants can reveal a fraction of their keys and see if the mismatch probability is suspiciously high. However, that would require the fraction of the keys revealed to be large enough for a precise estimation, as variance can mislead the participants in their decision regarding the presence of an internal or external cheater.

\subsection{Denoising QSS}

We have seen in the previous section that noise is a serious limitation to QSS protocols. In this section, we show how QAEs can limit the impact of noise by drastically diminishing the probability of the protocol failing.
\medskip

To do so, Charlie can apply the QAE to denoise the triplet before it sends Alice and Bob their qubits. Alternatively, Charlie can directly use the QAE to create an almost perfect GHZ triplet, as demonstrated in section \ref{Results}. This will dampen, if not almost annihilate, the probability of having erroneous shared bits between Alice and Bob, and Charlie. 

The circuit used for this experiment is represented in Figure \ref{fig:denoised_QSS}. We use our $[3,1,3]$ QAE trained on the QDC. Qubit 0 is Alice's qubit, qubit 1 is Bob's qubit, qubit 2 is Charlie's qubit, and qubit 3 is an ancilla used by the QAE (remember that the width of our QAEs is defined by $\max\limits_{1 \leq i \leq \ell - 1}(m_i + m_{i+1}) = 4$ here). At the output of the QAE, the state is encoded in qubit 0, 1 and 2 (but that may not always be the case, for instance if the QAE had shape $[3,1,1,3]$, the state would be encoded in qubits 1, 2 and 3).

\begin{figure*}
\includegraphics[scale = 0.34]{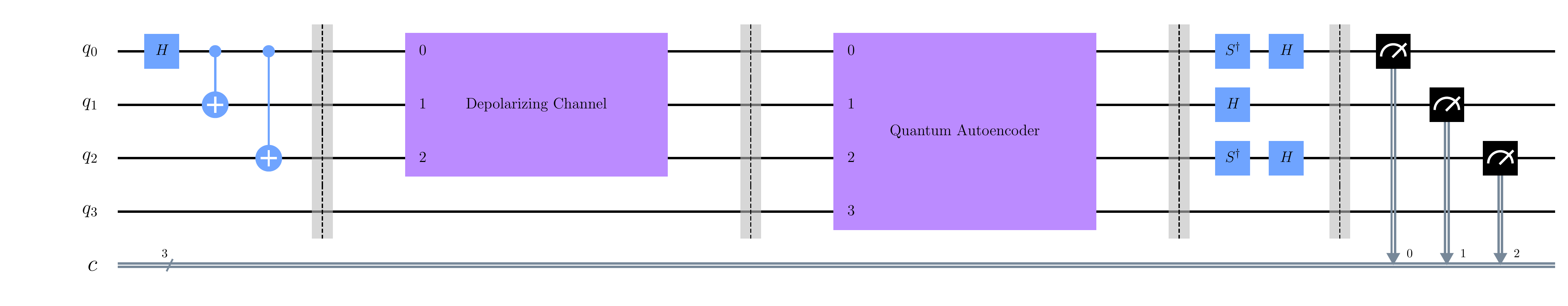}
\caption{Quantum circuit to denoise the QSS protocol}
\label{fig:denoised_QSS}
\end{figure*}

The results are shown in Figure \ref{fig:QSS_noisy}, for a noise strength ranging from 0 to 1. As one can see, the QAE manages to substantially reduce the failure probability regardless of the noise strength, hence saving the QSS protocol.

\section{Conclusion}
We successfully implemented the QAE from \cite{bondarenko2019quantum}, using Python and Qiskit. Additionally, we demonstrated the QAE denoising ability on GHZ states for multiple noise channels. The network is able to effectively denoise all of the given channels for various noise strengths, and is itself resistant to noise. Finally, we have shown that after being trained, the QAE is able to function as a generative model for the denoised state. It also proved itself particularly useful in the context of QSS. Other advantages of the QAE lie in its compactness and in the fact that it can be applied instantaneously once trained, despite a lengthy training process.

There are many avenues worth exploring with this QAE in the future. In particular, since the network can function as a quantum error correcting code (QECC), it would be interesting to benchmark it again other state-of-the-art QECCs, notably to compare its denoising accuracy as well as the time and resources involved.

In this paper, we were limited by computational resources, so we could only fully test $[2,1,2]$ and $[3,1,3]$ networks. In the future, it would be useful to explore the effectiveness of the network with different topologies.

\section*{Code}

The code to create, train and test the QAEs, as well as for the QSS part can be found at \url{https://github.com/Tom-Achache/QAEs}.

\bibliography{bib}
\bibliographystyle{ieeetr}

\end{document}